\documentclass[3p,letterpaper,twoside]{elsarticle}
\usepackage{amssymb,amsmath,float}

\DeclareMathOperator{\diag}{diag}
\renewcommand{\(}{\left(}
\renewcommand{\)}{\right)}
\newcommand{\defeq}{\stackrel{\Delta}{=}}
\newcommand{\Frac}[2]{{{#1}/{#2}}}
\newcommand{\ip}[2]{{\langle {#1},\, {#2} \rangle}}
\def\sign{\mathop{\mathrm{sign}}}
\def\R{{\mathbb{R}}}
\def\1{{\mathds{1}}}
\def\c{a}
\def\Half{\frac{1}{2}}
\def\half{\textstyle\frac{1}{2}}
\def\yInit{\hat{x}_{\rm init}}
\def\yinit{\hat{y}_{\rm init}}
\def\ts{\textstyle}

\def\pseudo{(F^{*}F)^{-1}F^{*}}
\def\id{{I}}
\def\setI{{\mathcal{I}}}
\def\setJ{{\mathcal{J}}}
\def\O{{P}}
\def\ccol{{J}}
\def\DI{D_{\rm I}}
\def\DII{D_{\rm II}}
\def\LI{L_{\rm I}}
\def\LII{L_{\rm II}}
\def\RI{R_{\rm I}}
\def\RII{R_{\rm II}}
\newcommand{\minus}[2]{{{D_{#1}^{(#2)}}}}
\newcommand{\Minus}[2]{{{\widetilde{D}_{#1}^{(#2)}}}}
\newtheorem{thm}{Theorem}[section]
\newtheorem{lemma}[thm]{Lemma}
\newtheorem{prop}[thm]{Proposition}

\newtheorem{coro}[thm]{Corollary}
\newtheorem{defi}[thm]{Definition}
\newproof{proof}{Proof}
\newlength{\widthA} 
\newlength{\widthB} 
\newlength{\xtraV}  
\floatstyle{ruled}
\newfloat{algorithm}{t}{loa}
\floatname{algorithm}{Algorithm}
\begin{document}

\title{Frame Permutation Quantization\tnoteref{t1,t2,t3}}
\tnotetext[t1]{The authors are with the Department of Electrical Engineering
  and Computer Science and the Research Laboratory of Electronics,
  Massachusetts Institute of Technology, Cambridge, MA 02139 USA\@.
  L. R. Varshney is additionally with the Laboratory for Information
  and Decision Systems, Massachusetts Institute of Technology.}
\tnotetext[t2]{This material is based upon work supported by the
  National Science Foundation under Grant No.\ 0729069\@.  This work was
  also supported in part by a Vietnam Education Foundation Fellowship.}
\tnotetext[t3]{This work was presented in part at the Inaugural Workshop of
  the Center for Information Theory and Its Applications at the
  University of California, San Diego, February 2006,
  and the Forty-fourth Annual Conference on Information Sciences and Systems,
  Princeton, NJ, March 2010.}
\author{Ha Q. Nguyen}
\author{Vivek K Goyal}
\author{Lav R. Varshney}

\begin{abstract}
Frame permutation quantization (FPQ) is a new vector quantization technique
using finite frames.  In FPQ, a vector is encoded using a permutation source
code to quantize its frame expansion.  This means that the encoding is a
partial ordering of the frame expansion coefficients.
Compared to ordinary permutation source coding,
FPQ produces a greater number of possible quantization
rates and a higher maximum rate.
Various representations for the partitions induced by FPQ are presented,
and reconstruction algorithms based on linear programming, quadratic
programming, and recursive orthogonal projection are derived.
Implementations of the linear and quadratic programming algorithms
for uniform and Gaussian sources show performance improvements
over entropy-constrained scalar quantization for certain combinations of
vector dimension and coding rate.
Monte Carlo evaluation of the recursive algorithm shows
that mean-squared error (MSE) decays as $M^{-4}$ for an $M$-element frame,
which is consistent with previous results on optimal decay of MSE\@.
Reconstruction using the canonical dual frame is also studied,
and several results relate properties of the analysis frame to
whether linear reconstruction techniques provide consistent reconstructions.
\end{abstract}


\maketitle
\pagestyle{empty}
\thispagestyle{empty}

\section{Introduction}
Redundant representations obtained with frames are playing an
ever-expanding role in signal processing due to design flexibility
and other desirable properties~\cite{KovacevicC:07a,KovacevicC:07b}.
One such favorable property is robustness to
additive noise~\cite{Daubechies1992}.
This robustness, carried over to quantization noise
(without regard to whether it is random or signal-independent),
explains the success of both ordinary oversampled analog-to-digital
conversion (ADC) and $\Sigma$--$\Delta$ ADC\@
with the canonical linear reconstruction.
But the combination of frame expansions with scalar quantization
is considerably more interesting and intricate
because boundedness of quantization noise can be exploited in
reconstruction~\cite{ThaoV1994,ThaoV:94b,GoyalVT1998,Cvetkovic:03,RanganG:01,BenedettoPY:06,BodmannP:07,BodmannL:08,Powell:10}
and frames and quantizers can be designed jointly to obtain favorable
performance~\cite{BeferullLozanoO:03}.

This paper introduces a new use of finite frames in vector quantization:
\emph{frame permutation quantization} (FPQ).
In FPQ, permutation source coding (PSC)~\cite{Dunn:65,BergerJW:72}
is applied to a frame expansion of a vector.
This means that the vector is represented by a partial ordering of
the frame coefficients (Variant I) or by signs of the frame coefficients
that are larger than some threshold along with a partial ordering of
the absolute values of the significant coefficients (Variant II).
FPQ provides a space partitioning that can be combined with
additional signal constraints or prior knowledge to generate a variety
of vector quantizers.

Beyond the explication of the basic ideas in FPQ,
the focus of this paper is on how---in analogy to works cited above---there
are several decoding procedures that can sensibly be used with the
encoding of FPQ\@.  First, we consider using the ordinary PSC decoding for
the frame coefficients followed by linear synthesis with the canonical dual;
from the perspective of frame theory, this is the natural way to reconstruct.
For this, we find conditions on the frame used in FPQ that
relate to whether the canonical reconstruction is consistent.
Second, taking a geometric approach based on imposing consistency yields
instead optimization-based algorithms.
Third, algorithms with lower complexity can have similar performance by
recursively imposing consistency only locally~\cite{RanganG:01,Powell:10}.

There are two distinct ways to measure the performance of FPQ,
and these correspond to different potential uses of FPQ:
data compression and data acquisition.
The accuracy of signal representation---here measured by mean-squared error
(MSE)---is important in either case. 
For data compression, accuracy is traded off against a coding rate
(bits per sample).  The standard alternative is scalar quantization, and
for low delay and complexity, one considers moderate signal dimensions.
It is remarkable that introducing redundancy through a frame expansion
can improve compression, and we find that it does so only when the redundancy
is low.
For data acquisition, accuracy is traded off against the number of
samples collected (number of frame elements).
Sensors that operate at low power and high speed by outputting orderings
of signal levels rather than absolute levels have been demonstrated and
are a subject of renewed interest~\cite{BrajovicK:99,GuoQH:07}.
By showing that the MSE can decay quickly as a function of the number
of samples collected, we may encourage the further development of
such sensors.
Here computational complexity of reconstruction is more important
because the data are recoded prior to storage or transmission.
This is in close analogy to oversampling in analog-to-digital conversion,
which is ubiquitous even though it is not advantageous in terms of accuracy
as a function of bit rate unless there is recoding at or near
Nyquist rate~\cite{CvetkovicV:98-IT,CvetkovicDL:07}.
Note also that for both historical and practical reasons,
data compression is typically studied for random vectors while
data acquisition is studied for nonrandom vectors within
some bounded set~\cite{DonohoVDD:98}.
This paper mixes Bayesian and non-Bayesian formulations accordingly.

The paper is organized as follows:
Before formal introduction to frame expansions, permutation source codes,
or their combination, Section~\ref{sec:preview} provides a preview of
the geometry of FPQ\@.  This serves both to contrast with ordinary
scalar-quantized frame expansions and to see the effect of frame redundancy.
Section~\ref{sec:background} provides the requisite background by reviewing
PSCs, frames, and scalar-quantized frame expansions.
Section~\ref{sec:framePSC} formally defines FPQ,
emphasizing constraints that are implied by the representation and hence
must be satisfied for consistent reconstruction.
Section~\ref{sec:framePSC} also provides reconstruction algorithms
based on applying the constraints for consistent reconstruction
globally or locally.
The results on choices of frames in FPQ appear in Section~\ref{sec:consistency}.
These are necessary and sufficient conditions on frames for
linear reconstructions to be consistent.
Section~\ref{sec:simulations} provides numerical results that demonstrate
improvement in operational distortion--rate compared to ordinary PSC
and optimal decay of distortion as a function of the number of samples.
Proofs of the main results are given in Section~\ref{sec:proofs}.
Preliminary results on FPQ were mentioned briefly in~\cite{VarshneyG:06a}.

\section{Preview through $\R^2$ Geometry}
\label{sec:preview}

Consider the quantization of $x \in \R^N$, where we restrict attention to $N=2$
in this section but later allow any finite $N$.
The uniform scalar quantization of $x$ partitions $\R^N$ in a trivial way,
as shown in Fig.~\ref{fig:qfe}(a).
(An arbitrary segment of the plane is shown.)
If over a domain of interest each component is divided into $K$ intervals,
a partition with $K^N$ cells is obtained.

\begin{figure}
 \begin{center}
  \begin{tabular}{ccc}
   \includegraphics[width=\widthB]{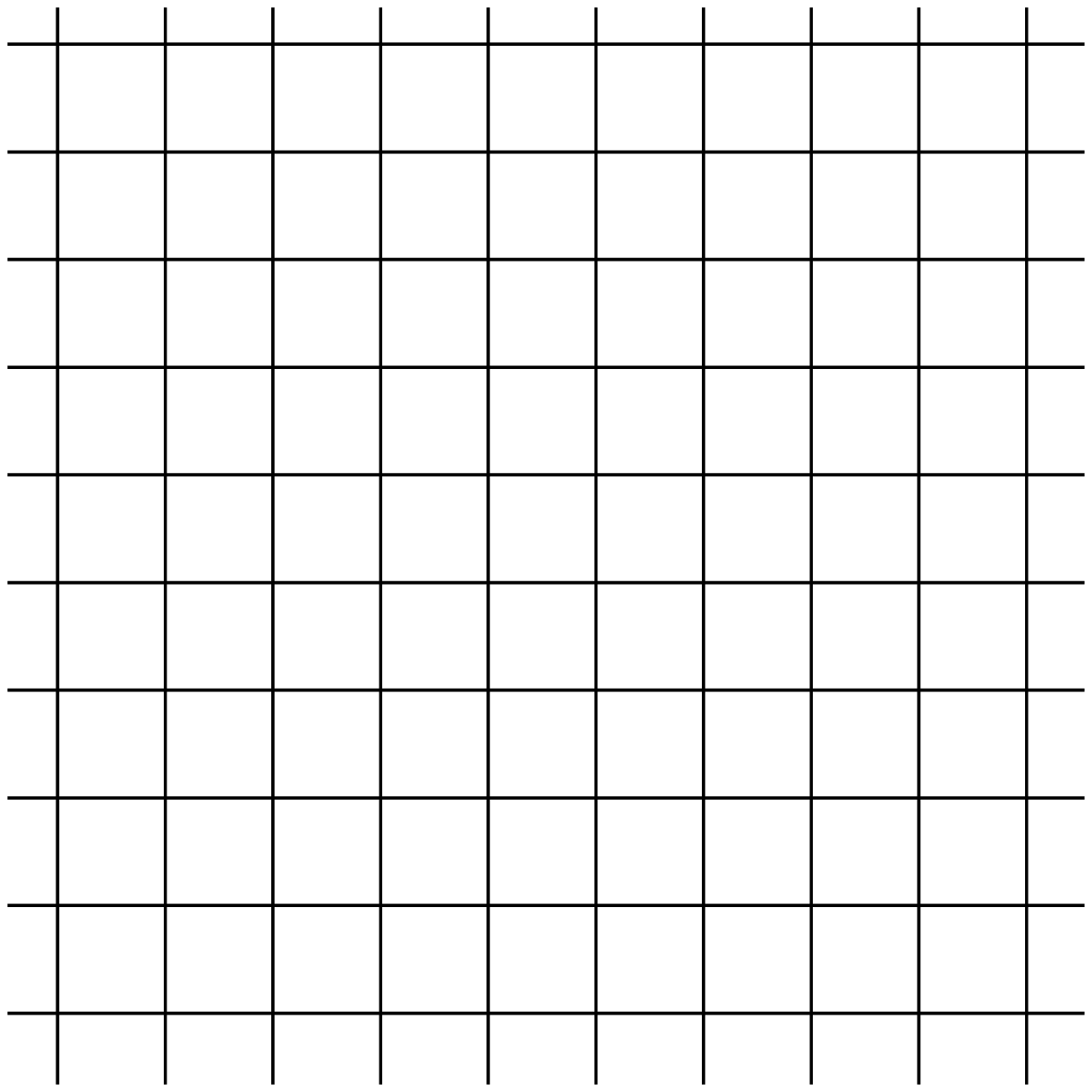} &
   \includegraphics[width=\widthB]{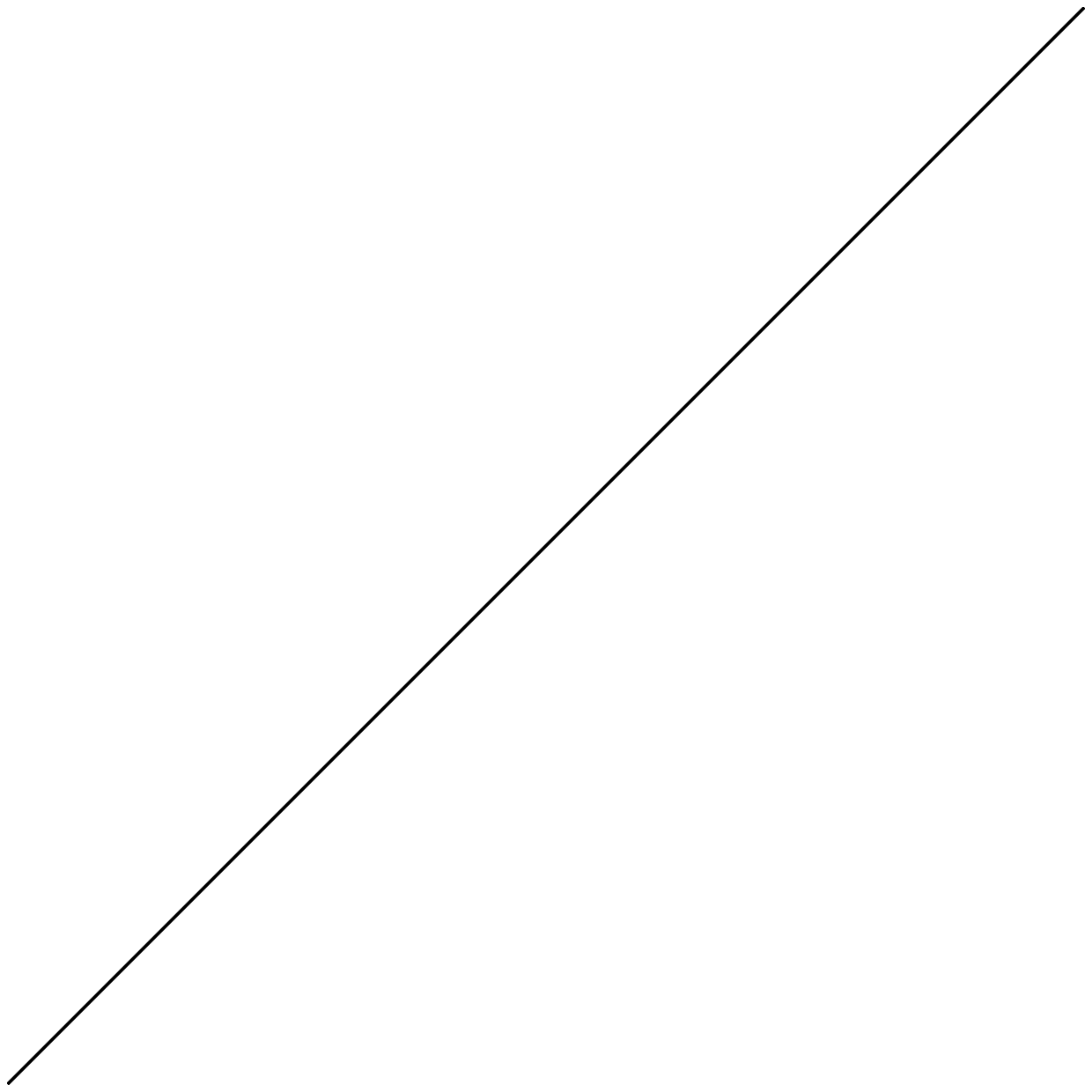} &
   \includegraphics[width=\widthB]{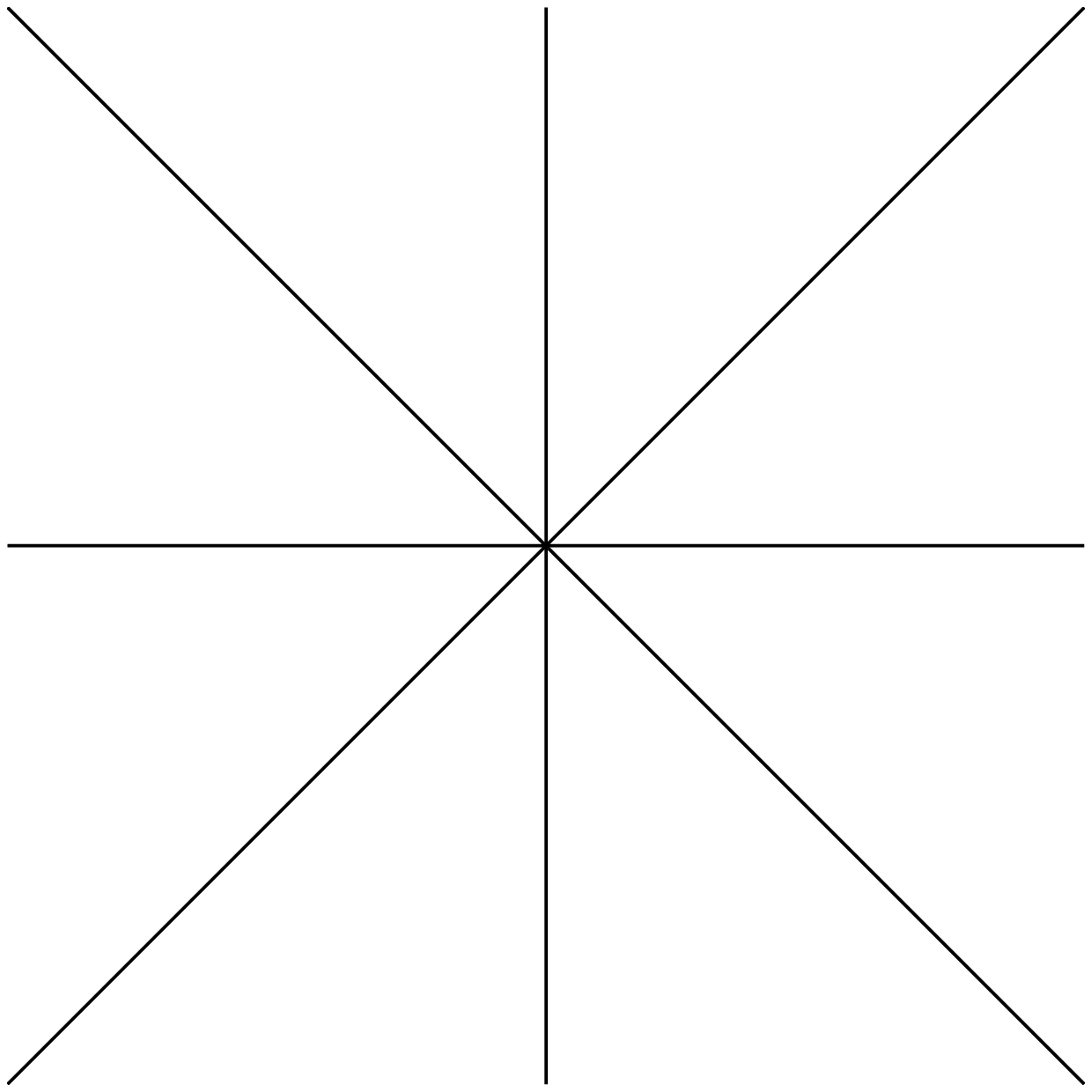} \\
   {\scriptsize (a) Scalar quantization} &
   {\scriptsize (b) Permutation source code (Var.\ I)} &
   {\scriptsize (c) Permutation source code (Var.\ II)} \\
   \includegraphics[width=\widthB]{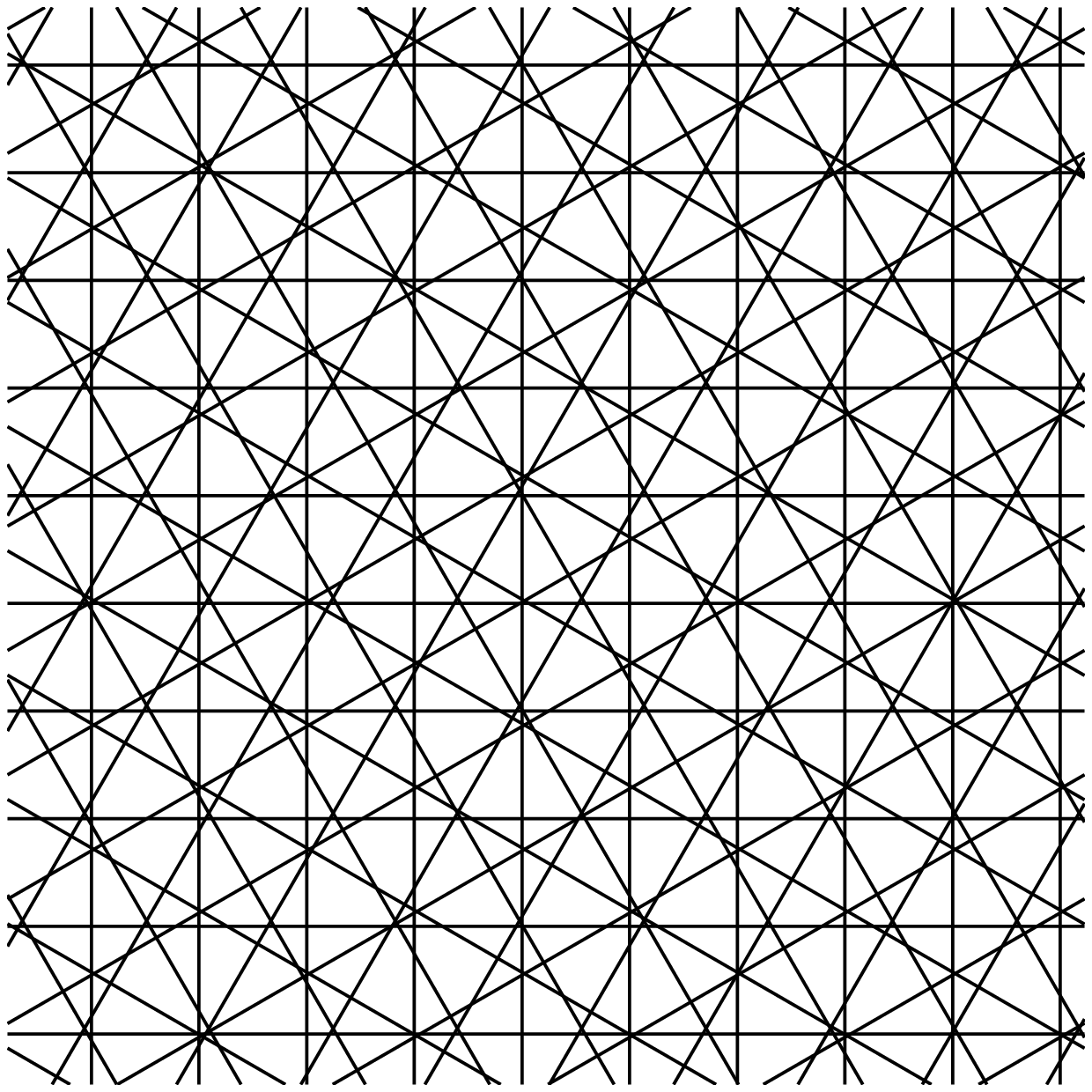} &
   \includegraphics[width=\widthB]{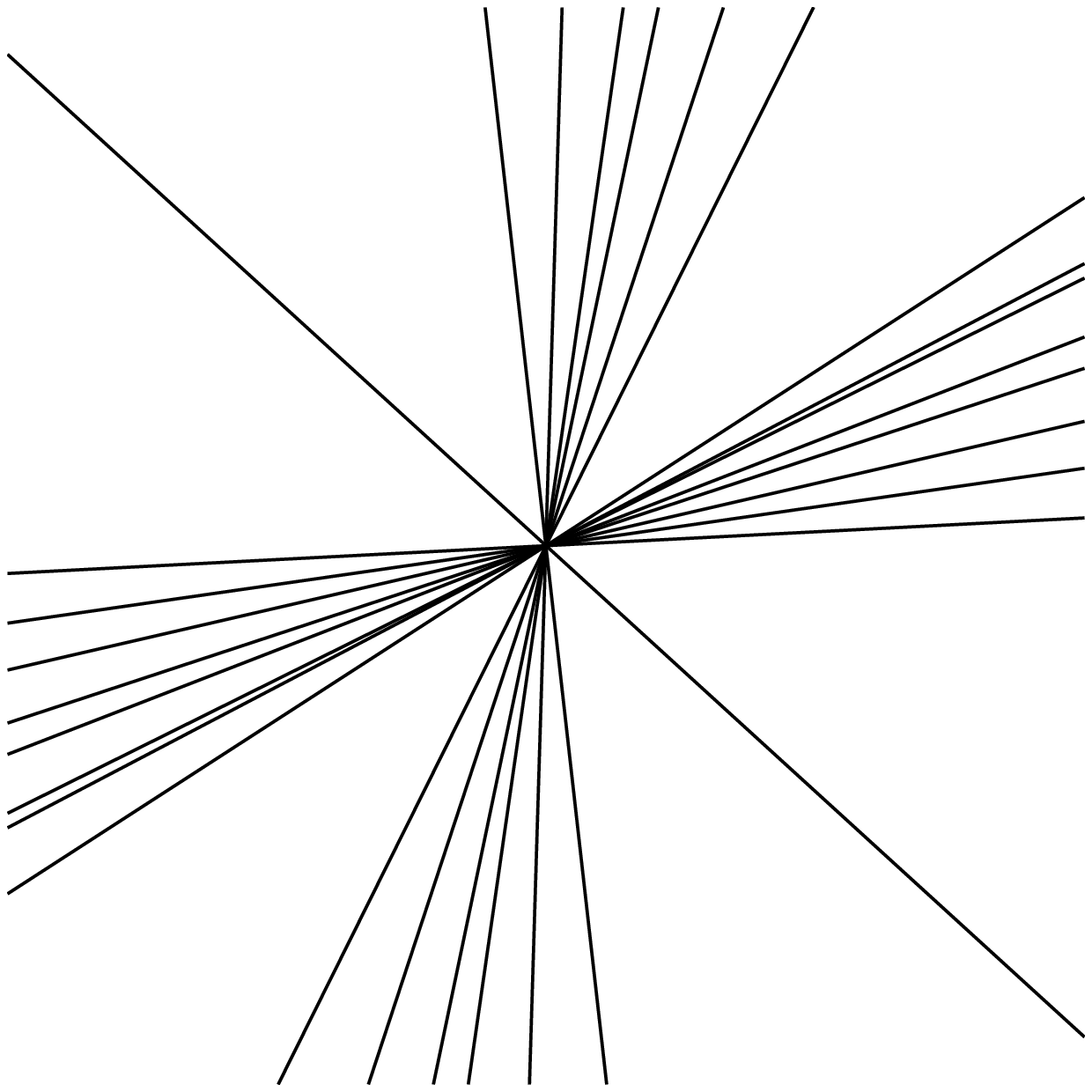} &
   \includegraphics[width=\widthB]{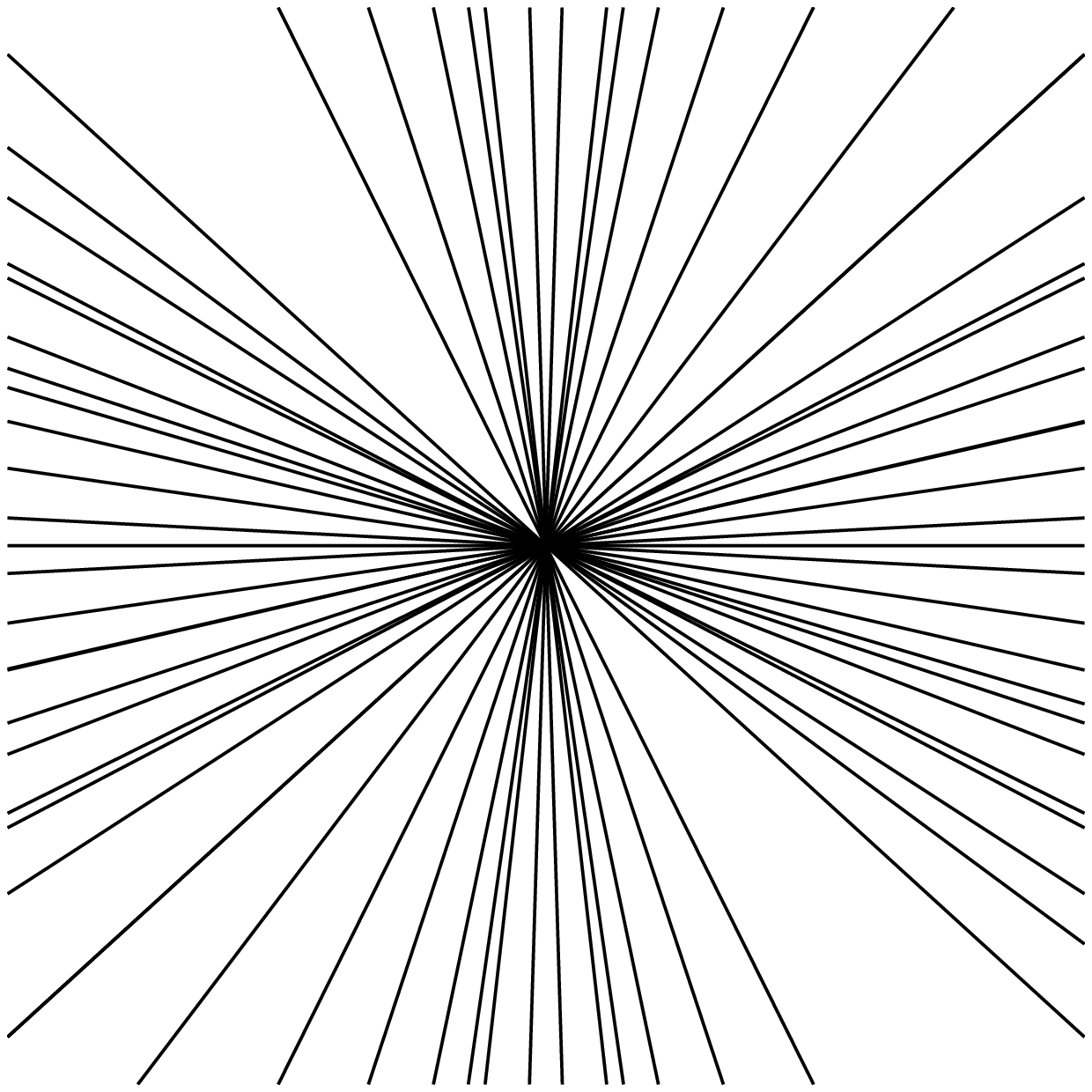} \\
   {\scriptsize (d) Scalar-quantized frame expansion} &
   {\scriptsize (e) Frame permutation quantizer (Var.\ I)} &
   {\scriptsize (f) Frame permutation quantizer (Var.\ II)}
  \end{tabular}
 \end{center}
 \caption{Partition diagrams for $x \in \R^2$.
   (a) Scalar quantization.
   (b) Permutation source code, Variant I\@. 
   (c) Permutation source code, Variant II\@. 
   (Both permutation source codes have $n_1=n_2=1$.)
   (d) Scalar-quantized frame expansion with $M=6$ coefficients
       (real harmonic tight frame).
   (e) Frame permutation quantizer, Variant I\@.
   (f) Frame permutation quantizer, Variant II\@.
   (Both frame permutation quantizers have $M=6$,
   $m_1=m_2=\cdots=m_6=1$, and the same random frame.)
  }
 \label{fig:qfe}
\end{figure}

A way to increase the number of partition cells without increasing the
scalar quantization resolution is to use a frame expansion.
A conventional quantized frame expansion is obtained by scalar quantization
of $y = Fx$, where $F \in \R^{M \times N}$ with $M \geq N$.
Keeping the resolution $K$ fixed, the partition now has $K^M$ cells.
An example with $M=6$ is shown in Fig.~\ref{fig:qfe}(d).
Each frame element $\phi_k$ (transpose of row of $F$)
induces a \emph{hyperplane wave partition}~\cite{ThaoV:96}:
a partition formed by equally-spaced $(N-1)$-dimensional
hyperplanes normal to $\phi_k$.  The overall partition has
$M$ hyperplane waves and is spatially uniform.  A spatial shift
invariance can be ensured formally by the use of subtractively
dithered quantizers~\cite{GrayS:93}.

A Variant I PSC represents $x$ just by which permutation
of the components of $x$ puts the components in descending order.
In other words,
only whether $x_1 > x_2$ or $x_2 > x_1$ is specified.\footnote{The
boundary case of $x_1 = x_2$ can be handled arbitrarily in practice
and safely ignored in the analysis.  When the source vector has an
absolutely continuous distribution, the boundary affects neither the
rate nor the distortion.  For an optimal quantizer, the boundaries
will have zero probability even if the source has
a discrete component~\cite[p.~355]{GershoG:92}.}
The resulting partition is shown in Fig.~\ref{fig:qfe}(b).
A Variant II PSC specifies (at most) the signs of the components
of $x_1$ and $x_2$ and whether $|x_1| > |x_2|$ or $|x_2| > |x_1|$.
The corresponding partitioning of the plane is shown in Fig.~\ref{fig:qfe}(c),
with the vertical line coming from the sign of $x_1$,
the horizontal line coming from the sign of $x_2$,
and the diagonal lines from $|x_1| \gtrless |x_2|$.

While low-dimensional diagrams are often inadequate in explaining PSC,
several key properties are illustrated.
The partition cells are (unbounded) convex cones,
giving special significance to the origin and a lack of
spatial shift invariance.
The unboundedness of cells implies that some additional knowledge,
such as a bound on $\| x \|$ or a probabilistic distribution on $x$,
is needed to compute good estimates.
At first this may seem extremely different from ordinary
scalar quantization or scalar-quantized frame expansions, but those techniques
also require some prior knowledge to allow the quantizer outputs to be
represented with finite numbers of bits.
We also see that the dimension $N$ determines the maximum number of
cells ($N!$ for Variant I and $2^N N!$ for Variant II);
there is no parameter analogous to scalar quantization step size that
allows arbitrary control of the resolution.

To get a finer partition without changing the dimension $N$,
we can again employ a frame expansion.
With $y = Fx$ as before, PSC of $y$ gives more
relative orderings with which to represent $x$.
If $\phi_j$ and $\phi_k$ are frame elements
(transposes of rows of $F$)
then $\ip{x}{\phi_j} \gtrless \ip{x}{\phi_k}$
is $\ip{x}{\phi_j-\phi_k} \gtrless 0$ by linearity of the inner product,
so every pair of frame elements can give a condition on $x$.
An example of a partition obtained with Variant I and $M=6$ is shown in
Fig.~\ref{fig:qfe}(e).
There are many more cells than in Fig.~\ref{fig:qfe}(b).
Similarly, Fig.~\ref{fig:qfe}(f) shows a Variant II example.
The cells are still (unbounded) convex cones.
If additional information such as $\|x\|$ or an affine subspace constraint
(not passing through the origin) is known, $x$ can be specified arbitrarily
closely by increasing $M$.

\section{Background}
\label{sec:background}
Having illustrated the basic idea of PSC and our generalization
using frames to provide resolution control, we now formalize
the background material.
We assume throughout fixed-rate coding and
the conventional squared-error fidelity criterion
$\| x - \hat{x} \|^2$ between source $x$ and reproduction $\hat{x}$.
Some statements---especially those pertaining to data compression---assume
a known source distribution over which
performance is measured in expectation.
Most statements for data acquisition with $M \rightarrow \infty$ apply
pointwise over $x$.

\subsection{Vector Quantization}
A vector quantizer is a mapping from an input $x \in \R^N$ to
a \emph{codeword} $\hat{x}$ from a finite \emph{codebook} $\mathcal{C}$.
Without loss of generality, a vector quantizer can be seen as the
combination of an \emph{encoder}
$$
  \alpha: \R^N \rightarrow \mathcal{I}
$$
and a \emph{decoder}
$$
  \beta: \mathcal{I} \rightarrow \R^N,
$$
where $\mathcal{I}$ is a finite index set.
The encoder partitions $\R^N$ into $|\mathcal{I}|$ regions or \emph{cells}
$\{ \alpha^{-1}(i) \}_{i \in \mathcal{I}}$,
and the decoder assigns a \emph{reproduction value} to each cell.
Examples of partitions are given in Fig.~\ref{fig:qfe}.
For the quantizer to output $R$ bits per component,
we have $|\mathcal{I}| = 2^{NR}$.

For any codebook (i.e., any $\beta)$, the encoder $\alpha$ that minimizes
$\| x - \hat{x} \|^2$ maps $x$ to the nearest element of the codebook.
The partition is thus composed of convex cells.
Since the cells are convex, reproduction values are optimally within the
corresponding cells---whether to minimize mean-squared error distortion,
maximum squared error, or some other reasonable function of squared error.
To minimize maximum squared error, reproduction values should be at centers
of cells; to minimize expected distortion, they should be at centroids
of cells.
Reproduction values being within corresponding cells is formalized
as \emph{consistency}:
\begin{defi}
The reconstruction $\hat{x} = \beta(\alpha(x))$ is called a
\emph{consistent reconstruction of $x$} when $\alpha(x) = \alpha(\hat{x})$
(or equivalently $\beta(\alpha(\hat{x})) = \hat{x}$).
The decoder $\beta$ is called \emph{consistent} when $\beta(\alpha(x))$
is a consistent reconstruction of $x$ for all $x$.
\end{defi}

In practice, the pair $(\alpha,\beta)$ usually does not minimize any
desired distortion criterion for a given codebook size because the
optimal mappings are hard to design and hard to implement~\cite{GershoG:92}.
The mappings are commonly designed subject to certain structural constraints,
and $\beta$ may not even be consistent for $\alpha$~\cite{ThaoV1994,GoyalVT1998}.

\subsection{Permutation Source Codes}
\label{sec:permutationcodes}
A permutation source code is a vector quantizer with the defining
characteristic that codewords are related through permutations and,
possibly, sign changes.
Permutation codes were originally introduced as channel codes by
Slepian~\cite{Slepian1965}.  They were then applied to a specific source coding
problem, through the duality between source encoding and channel decoding,
by Dunn~\cite{Dunn:65} and developed in greater generality by
Berger \emph{et al.}~\cite{BergerJW:72,Berger:72,Berger:82}.
Permutation codes are generated by the group action of a permutation group
and are thus examples of group codes~\cite{Slepian:68}.

\subsubsection{Definitions}
There are two variants of permutation codes:

\emph{Variant I}:
Here codewords are related through permutations, without sign changes.
Let $\mu_1 > \mu_2 > \cdots > \mu_{K}$ be real numbers, and
let $n_1,n_2,\ldots,n_{K}$ be positive integers that sum to $N$
(an \emph{(ordered) composition} of $N$).
The \emph{initial codeword} of the codebook $\mathcal{C}$ has the form
\begin{equation}
    \yInit=(\mathop{{\mu_1,\ldots,\mu_1}}_{\longleftarrow n_1 \longrightarrow},\mathop{{\mu_2,\ldots,\mu_2}}_{\longleftarrow n_2 \longrightarrow},\ldots,\mathop{{\mu_K,\ldots,\mu_K}}_{\longleftarrow n_K \longrightarrow}),
\label{eq:init}
\end{equation}
where each $\mu_i$ appears $n_i$ times.
When $\yInit$ has this form, we call it \emph{compatible with
$(n_1,n_2,\ldots,n_K)$}.
The codebook is the set of all distinct permutations of $\yInit$.
The number of codewords in $\mathcal{C}$ is thus given by the multinomial coefficient
\begin{subequations}
\label{eq:size12}
\begin{equation}
\LI = \frac{N!}{n_1 ! \, n_2 ! \, \cdots \, n_{K} !}. \label{eq:size1}
\end{equation}

The permutation structure of the codebook enables low-complexity
nearest-neighbor encoding~\cite{BergerJW:72}: map $x$ to the codeword
$\hat{x}$ whose components have the same order as $x$;
in other words, replace the $n_1$ largest components of $x$ with $\mu_1$,
the $n_2$ next-largest components of $x$ with $\mu_2$, and so on.

\emph{Variant II}:
Here codewords are related through permutations and sign changes.
Let $\mu_1 > \mu_2 > \cdots > \mu_{K}\geq 0$ be nonnegative real numbers, and
let $\(n_1,n_2,\ldots,n_{K}\)$ be a composition of $N$.
The initial codeword has the same form as in (\ref{eq:init}),
and the codebook now consists of all distinct permutations of $\yInit$
with each possible sign for each nonzero component.
The number of codewords in $\mathcal{C}$ is thus given by
\begin{equation}
\LII = 2^{h}\frac{N!}{n_1 ! \, n_2 ! \, \cdots \, n_{K} !},\label{eq:size2}
\end{equation}
\end{subequations}
where $h=N$ if $\mu_K>0$ and $h=N-n_K$ if $\mu_K=0$.

Nearest-neighbor encoding for Variant II PSCs can be implemented
as follows~\cite{BergerJW:72}: map $x$ to the codeword
$\hat{x}$ whose components have the same order in absolute value and
match the signs of corresponding components of $x$.
Since the complexity of sorting a vector of length $N$ is $O(N \log N)$ operations,
the encoding complexity for either PSC variant is much lower than with
an unstructured
source code and only $O(\log N)$ times higher than scalar quantization.

With the codebook sizes given in (\ref{eq:size12}),
the per-component rate is defined as
\begin{equation}
R = N^{-1}\log_2 L. \label{eq:PSCrate}
\end{equation}
Under certain symmetry conditions on the source distribution,
all codewords are equally likely so the rate
cannot be reduced by entropy coding.
This generation of fixed-rate output---avoiding the possibility of buffer
overflow associated with entropy coding of the highly nonequiprobable
outputs of a quantizer~\cite{Jelinek1968}---is a known advantage of
PSCs~\cite{BergerJW:72}.
An efficient enumeration of permutations, to generate a binary representation,
is described in~\cite{Cover:73}.

\subsubsection{Partition Properties}
\label{sec:psc-partition}
For both historical reasons and to match the conventional approach
to vector quantization, PSCs were defined above in terms of a codebook structure,
and the codebook structure led to an encoding procedure.
Note that we may now examine the partitions induced by PSCs separately
from the particular codebooks for which they are nearest-neighbor partitions.

The partition induced by a Variant I PSC is completely determined by
the composition $\(n_1,n_2,\ldots,n_{K}\)$.
Specifically, the encoding mapping can index the permutation $P$ that
places the $n_1$ largest components of $x$ in the first $n_1$ positions
(without changing the order within those $n_1$ components),
the $n_2$ next-largest components of $x$ in the next $n_2$ positions,
and so on;
the $\mu_i$s are actually immaterial.
This encoding is placing all source vectors $x$ such that $Px$
is $n$-descending in the same partition cell, defined as follows.
\begin{defi}
Given a composition $n = (n_1,n_2,\ldots,n_K)$ of $N$,
a vector in $\R^N$ is called \emph{$n$-descending} if its $n_1$ largest
entries are in the first $n_1$ positions, its $n_2$ next-largest
components are in the next $n_2$ positions, etc.
\end{defi}
The property of being $n$-descending is to be descending up to the
arbitrariness specified by the composition $n$.

Because this is nearest-neighbor encoding for \emph{some} codebook,
the partition cells must be convex.
Furthermore, multiplying $x$ by any nonnegative scalar does not affect
the encoding, so the cells are convex cones.
(This was discussed and illustrated in Section~\ref{sec:preview}.)
We develop a convenient representation for the partition in
Section~\ref{sec:framePSC}.

The situation is only slightly more complicated for Variant II PSCs.
The partition is determined by the composition $\(n_1,n_2,\ldots,n_{K}\)$
and whether or not the signs of the smallest-magnitude components
should be encoded (whether $\mu_K = 0$, in the codebook-centric view).

The PSC literature has mostly emphasized the design of PSCs for
sources with i.i.d.\ components.
But as developed in Section~\ref{sec:framePSC}, the simple structured
encoding of PSCs could be combined with unconventional decoding techniques
for other sources.
The possible suitability of PSCs for sources with unknown or time-varying
statistics has been previously observed~\cite{BergerJW:72}.

\subsubsection{Codebook Optimization}
With the encoding procedure now fixed, let us turn to the decoder
(or codebook) design.
For this we assume that $x$ is random and that the components of $x$ are i.i.d.

Let $\xi_1\geq\xi_2\geq\cdots\geq\xi_N$
denote the order statistics of random vector $x=\(x_1,\ldots,x_N\)$ and
$\eta_1\geq\eta_2\geq\cdots\geq\eta_N$
denote the order statistics of random vector $|x|\defeq\(|x_1|,\ldots,|x_N|\)$.%
\footnote{For consistency with earlier literature on PSCs, we are reversing the
usual sorting of order statistics~\cite{DavidN:03}.}
For a given initial codeword $\yInit$,
the per-letter distortion of optimally-encoded Variant I and Variant II PSCs are given by
\begin{subequations}
\begin{equation}
\label{eq:distort}
\DI = N^{-1}E\left[{\sum_{i=1}^{K}\sum_{\ell \in \setI_i}\left(\xi_\ell-\mu_i\right)^2}\right]
\end{equation}
and
\begin{equation}
\label{eq:distort2}
\DII = N^{-1}E\left[{\sum_{i=1}^{K}\sum_{\ell \in \setI_i}\left(\eta_\ell-\mu_i\right)^2}\right] ,
\end{equation}
\end{subequations}
where $\setI_i$s are the sets of indexes generated by the composition:
\begin{subequations}
\label{eq:index-sets}
\begin{eqnarray}
\setI_1 & = & \{1,2, \ldots ,n_1 \}, \\
\setI_i & = & \left\{\left({\textstyle\sum_{k=1}^{i-1}n_k}\right)+1,\,
                 \left({\textstyle\sum_{k=1}^{i-1}n_k}\right)+2,\,\ldots,\,
                 \left({\textstyle\sum_{k=1}^{i}n_k}\right)\right\}
  , \qquad i\geq 2.
\end{eqnarray}
\end{subequations}
These distortions can be deduced simply by examining which components of
$x$ are mapped to which elements of $\yInit$.

Optimization of (\ref{eq:distort}) and (\ref{eq:distort2}) over both
$\{n_i\}_{i=1}^K$ and $\{\mu_i\}_{i=1}^K$ subject to~(\ref{eq:PSCrate})
is difficult, partly due to the integer constraint of the composition.
However, given a composition $(n_1,n_2,\ldots,n_K)$,
the optimal initial codeword can be determined easily from the means
of the order statistics. In particular, the optimal $\ts\{\mu_i\}_{i=1}^K$
of Variant I and Variant II PSCs are given by
\begin{subequations}
\begin{equation}
\label{eq:optInit1}
\mu_i=n_i^{-1}\sum_{\ell\in \setI_i}E\left[\xi_\ell\right], \qquad \mbox{for Variant I},
\end{equation}
and
\begin{equation}
\label{eq:optInit2}
\mu_i=n_i^{-1}\sum_{\ell\in \setI_i}E\left[\eta_\ell\right], \qquad \mbox{for Variant II}.
\end{equation}
\end{subequations}

The analysis of \cite{Berger:72} shows that when $N$ is large,
the optimal composition gives performance equal to
optimal entropy-constrained scalar quantization (ECSQ) of $x$.
Performance does not strictly improve with increasing $N$;
permutation codes outperform ECSQ for certain combinations of block size
and rate~\cite{GoyalSW:01}.

\subsection{Frame Definitions and Classifications}
The theory of finite-dimensional frames is often developed for a Hilbert space
$\mathbb{C}^N$ of complex vectors. In this paper, we use frame expansions only
for quantization using PSCs, which rely on order relations of real numbers.
Therefore we limit ourselves to real finite frames.
We maintain the Hermitian transpose notation~$^*$ where a transpose would
suffice because this makes several expressions have familiar appearances.

The Hilbert space of interest is $\R^N$ equipped with
the standard inner product (dot product),
\[
\ip{x}{y} = x^{T}y = \sum_{k=1}^N x_k y_k,
\]
for $x=\left[x_1,x_2,\ldots,x_N\right]^T\in\R^N$ and
$y=\left[y_1,y_2,\ldots,y_N\right]^T\in\R^N$.
The norm of a vector $x$ is naturally induced from the inner product,
\[
\|x\|=\sqrt{\ip{x}{x}}.
\]

\begin{defi}[\cite{Daubechies1992}]
A set of $N$-dimensional vectors,
$\Phi=\{\phi_k\}_{k=1}^M\subset\mathbb{R}^{N}$,
is called a \emph{frame} if there exist a lower frame bound, $A>0$,
and an upper frame bound, $B <\infty$,
such that
\begin{subequations}
\label{eq:frameDef}
\begin{equation}
  A\|x\|^2\leq\sum_{k=1}^M|\ip{x}{\phi_k}|^2 \leq B\|x\|^2,
  \qquad
  \mbox{for all $x\in\R^{N}$}.\label{eq:frameDef1}
\end{equation}
The matrix $F \in \R^{M \times N}$ with $k$th row equal to $\phi_k^*$
is called the \emph{analysis frame operator}.
$F$ and $\Phi$ will be used interchangeably to refer to a frame.
Equivalent to (\ref{eq:frameDef1}) in matrix form is
\begin{equation}
A\id_N\leq F^{*}F \leq B\id_N,\label{eq:frameDef2}
\end{equation}
\end{subequations}
where $\id_N$ is the $N\times N$ identity matrix.
\end{defi}

The lower bound in~(\ref{eq:frameDef}) implies that $\Phi$ spans $\R^{N}$;
thus a frame must have $M\geq N$.
It is therefore reasonable to call the ratio $r=M/N$ the
\emph{redundancy} of the frame.  A frame is called a
\emph{tight frame} if the frame bounds can be chosen to be equal.
A frame is an \emph{equal-norm frame} if all of its vectors have the same norm.
If an equal-norm frame is normalized to have all vectors of unit norm, we call it a
\emph{unit-norm frame}
(or sometimes \emph{normalized frame} or \emph{uniform frame}).
For a unit-norm frame, it is easy to verify that $A \leq r \leq B$.
Thus, a unit-norm tight frame (UNTF) must satisfy $A = r = B$ and
\begin{equation}
F^{*}F = r \id_N.\label{eq:TF}
\end{equation}

Naimark's theorem~\cite{HanL2000} provides an efficient way to characterize
the class of equal-norm tight frames: a set of vectors is an equal-norm tight frame
if and only if it is the orthogonal projection (up to a scale factor)
of an orthonormal basis of an ambient Hilbert space on to some subspace.%
\footnote{The theorem holds for a general separable Hilbert space of possibly infinite dimension.}
As a consequence, deleting the last $(M-N)$ columns of the (normalized)
discrete Fourier Transform (DFT) matrix in $\mathbb{C}^{M\times M}$ yields
a particular subclass of UNTFs called
\emph{(complex) harmonic tight frames} (HTFs).
One can adapt this derivation to construct \emph{real} HTFs~\cite{GoyalKK:01},
which are always UNTFs, as follows.
\begin{defi}
\label{def:HTF}
The {real harmonic tight frame} of $M$ vectors in $\R^N$ is defined for even $N$ by
\begin{subequations}
\begin{align}
\phi^{*}_{k+1}=&\sqrt{\frac{2}{N}}\left[\cos{\frac{k\pi}{M}},\cos{\frac{3k\pi}{M}},\ldots,\cos{\frac{(N-1)k\pi}{M}},
\sin{\frac{k\pi}{M}},\sin{\frac{3k\pi}{M}}, \ldots,\sin{\frac{(N-1)k\pi}{M}}\right] \label{eq:HTFeven}
\end{align}
and for odd $N$ by
\begin{align}
\phi^{*}_{k+1}=&\sqrt{\frac{2}{N}}\left[\frac{1}{\sqrt{2}},\cos{\frac{2k\pi}{M}},\cos{\frac{4k\pi}{M}},\ldots,\cos{\frac{(N-1)k\pi}{M}},
\sin{\frac{2k\pi}{M}},\sin{\frac{4k\pi}{M}}, \ldots,\sin{\frac{(N-1)k\pi}{M}}\right],\label{eq:HTFodd}
\end{align}
\end{subequations}
where $k=0,1,\ldots,M-1$.
The \emph{modulated harmonic tight frames} are defined by
\begin{equation}
\label{eq:modulatedHTF}
\psi_k=\gamma(-1)^k\phi_k,
\qquad
\mbox{for $k=1,2,\ldots,M$},
\end{equation}
where $\gamma = 1$ or $\gamma = -1$ (fixed for all $k$).
\end{defi}

HTFs can be viewed as the result of a group of orthogonal operators acting
on one generating vector~\cite{KovacevicC:07b}. This property
has been generalized in~\cite{EldarB2003,Kolesar2004} under the
name \emph{geometrically-uniform frames} (GUFs).
Note that a GUF is a special case of a
group code as developed by Slepian~\cite{Slepian1965,Slepian:68}.
An interesting connection between PSCs and GUFs
is that under certain conditions, a PSC codebook is a GUF
with generating vector $\yInit$ and the generating group action
provided by all permutation matrices~\cite{Abdelkefi2008}.

Classification of frames is often up to some unitary
equivalence~\cite{GoyalKK:01}.
Adopting the terminology of Holmes and Paulsen~\cite{HolmesP2004},
we say two frames in $\R^N$,
$\Phi=\{\phi_k\}_{k=1}^M$ and $\Psi=\{\psi_k\}_{k=1}^M$,
are
\begin{enumerate}
\item [(i)] Type I equivalent if there is an orthogonal matrix $U$ such that $\psi_k=U\phi_k$ for all $k$;
\item [(ii)] Type II equivalent if there is a permutation $\sigma(\cdot)$ on
            $\{1,2,\ldots,M\}$ such that $\psi_k=\phi_{\sigma(k)}$ for all $k$; and
\item [(iii)] Type III equivalent if there is a \emph{sign function} in $k$,
            $\delta(k)=\pm 1$ such that $\psi_k=\delta(k)\phi_k$ for all $k$.
\end{enumerate}
It will be evident that FPQ performance is always invariant to Type II
equivalence; invariant to Type I equivalence when the source distribution
is rotationally invariant; and invariant to Type III equivalence under
Variant II but not, in general, under Variant I\@.

It is important to note
that for $M=N+1$ there is exactly one equivalence class of
UNTFs~\cite[Thm.~2.6]{GoyalKK:01}. Since HTFs are always UNTFs, the following
property follows directly from~\cite[Thm.~2.6]{GoyalKK:01}.
\begin{prop}\label{prop:equiv}
Assume that $M=N+1$, and $\Phi=\{\phi_k\}_{k=1}^{M}\subset\R^{N}$ is the real HTF\@.
Then every UNTF $\Psi=\left\{\psi_k\right\}_{k=1}^M$ can be written as
\begin{equation}
\psi_k=\delta(k) U\phi_{\sigma(k)},
\qquad
\mbox{for $k=1,2,\ldots,M$},
\end{equation}
where $\delta(k)=\pm 1$ is some sign function in $k$,
$U$ is some orthogonal matrix,
and $\sigma(\cdot)$ is some permutation on the index set $\{1,2,\ldots,M\}$.
\end{prop}

Another important subclass of UNTFs is defined as follows:
\begin{defi}[\cite{StrohmerH2003,SustikTDH2007}]
\label{def:ETF}
A UNTF $\Phi=\{\phi_k\}_{k=1}^M\subset\mathbb{R}^{N}$ is called an
\emph{equiangular tight frame} (ETF) if there exists a constant
$\c$ such that $|\ip{\phi_\ell}{\phi_k}|=\c$
for all $1\leq \ell<k\leq M$.
\end{defi}
ETFs are sometimes called \emph{optimal Grassmannian frames}
or \emph{$2$-uniform frames}.
They prove to have rich application in communications, coding theory,
and sparse approximation~\cite{HolmesP2004,StrohmerH2003,Tropp2004}.
For a general pair $(M,N)$,
the existence and constructions of such frames is not fully understood.
Partial answers can be found in~\cite{SustikTDH2007,Strohmer2008,CasazzaRT2008}.

In our analysis of FPQ, we will find that \emph{restricted} ETFs---where the
absolute value constraint can be removed from Definition~\ref{def:ETF}---play
a special role.  In matrix view, a restricted ETF satisfies $F^{*}F=r\id_N$ and
$FF^{*}=(1-\c)\id_M+\c J_M$, where $J_M$ is the all-1s matrix of size $M\times M$.
The following proposition specifies the restricted ETFs
for the codimension-1 case.
\begin{prop}\label{prop:ETF}
For $M=N+1$, the family of all restricted ETFs is constituted by the
Type~I and Type~II equivalents of modulated HTFs.
\end{prop}
\begin{proof}
See Section~\ref{sec:ETF}.
\end{proof}

The following property of modulated HTFs in the $M=N+1$ case will be very useful.
\begin{prop}
If $M=N+1$ then a modulated harmonic tight frame is a zero-sum frame, i.e.,
each column of the analysis frame operator $F$ sums to zero.
\end{prop}
\begin{proof}
We only consider the case when $N$ is even; the $N$ odd case is similar.
For each $\ell \in\{1,2,\ldots,N\}$,
let $\phi_k^\ell$ denote the $\ell$th component of vector $\phi_k$ and let
$S_\ell=\sum_{k=1}^M\phi_k^\ell$ denote the sum of the entries in column $\ell$
of matrix $F$.

For $1\leq \ell \leq N/2$, using Euler's formula, we have
\begin{eqnarray}
S_\ell&=&\pm\sqrt{\frac{2}{N}}\sum_{k=0}^{M-1}(-1)^k \cos{\frac{(2\ell-1)k\pi}{M}}
      \nonumber \\
    &\propto&\sum_{k=0}^{M-1}e^{jk\pi} \left[e^{j\Frac{(2\ell-1)k\pi}{M}}+e^{-j\Frac{(2\ell-1)k\pi}{M}}\right]
      \nonumber \\
    &=&\sum_{k=0}^{M-1}e^{j\pi\left(\Frac{(2\ell-1)}{M}+1\right)k}+\sum_{k=0}^{M-1}e^{-j\pi\left(\Frac{(2\ell-1)}{M}+1\right)k}
      \nonumber \\
    &=&\frac{1-e^{j\pi(2\ell+M-1)}}{1-e^{j\pi\left(\Frac{(2\ell-1)}{M}+1\right)}}+\frac{1-e^{-j\pi(2\ell+M-1)}}{1-e^{-j\pi\left(\Frac{(2\ell-1)}{M}+1\right)}}
      \nonumber \\
    &=&0, \label{eq:zero}
\end{eqnarray}
where (\ref{eq:zero}) follows from the fact that $2\ell+M-1=2\ell+N$ is an even integer.

For $N/2< \ell \leq N$, we can show that $S_\ell =0$ similarly, and so the proposition is proved.
\end{proof}

\subsection{Reconstruction from Frame Expansions}
A central use of frames is to formalize the reconstruction
of $x \in \R^N$ from the frame expansion
$y_k = \ip{x}{\phi_k}$, $k = 1, \, 2, \, \ldots, \, M$,
or estimation of $x$ from degraded versions of the frame expansion.
Using the analysis frame operator we have $y = Fx$,
and (\ref{eq:frameDef}) implies the existence of at least
one linear \emph{synthesis operator} $G$ such that $GF = I_N$.
A frame with analysis frame operator $G^*$ is then said to be \emph{dual}
to $\Phi$.

The frame condition (\ref{eq:frameDef}) also implies that $F^*F$ is invertible,
so 
the Moore--Penrose inverse (pseudo-inverse) of the frame operator
\[
F^\dagger=\pseudo
\]
exists and is a valid synthesis operator.
Using the pseudo-inverse for reconstruction has several important properties
including an optimality for mean-squared error (MSE)
under assumptions of uncorrelated zero-mean additive noise and
linear synthesis~\cite[Sect.~3.2]{Daubechies1992}.
This follows from the fact that $FF^{\dagger}$ is an orthogonal projection from $\R^{M}$ onto the subspace $F(\R^{N})$, the range of $F$.
Because of this special role, reconstruction using $F^\dagger$ is
called \emph{canonical reconstruction}
and the corresponding frame is called the \emph{canonical dual}.
In this paper, we use the term \emph{linear reconstruction}
for reconstruction using an arbitrary linear operator.

When $y$ is quantized to $\hat{y}$, it is possible for the quantization noise
$\hat{y} - y$ to have mean zero and uncorrelated components;
this occurs with subtractive dithered quantization~\cite{GrayS:93}
or under certain asymptotics~\cite{ViswanathanZ:01}.
In this case, the optimality of canonical reconstruction holds.
However, it should be noted that even with these restrictions,
canonical reconstruction is optimal \emph{only amongst linear reconstructions}.

When nonlinear construction is allowed, quantization noise
may behave fundamentally differently than other additive noise.
The key is that a quantized value gives hard constraints that can
be exploited in reconstruction.
For example, suppose that $\hat{y}$ is obtained from $y$ by
rounding each element to the nearest multiple of a quantization
step size $\Delta$.  Then knowledge of $\hat{y}_k$ is equivalent
to knowing
\begin{equation}
\label{eq:qfe-one-interval}
y_k \in [\hat{y}_k - \half \Delta, \, \hat{y}_k + \half \Delta].
\end{equation}
Geometrically,
$\ip{x}{\phi_k} = \hat{y}_k - \Half \Delta$
and
$\ip{x}{\phi_k} = \hat{y}_k + \Half \Delta$
are hyperplanes perpendicular to $\phi_k$,
and (\ref{eq:qfe-one-interval}) expresses that $x$ must lie
between these hyperplanes;
this may be visualized as one pair of parallel lines in Fig.~\ref{fig:qfe}(d).
Using the upper and lower bounds on all $M$ components of $y$,
the constraints on $x$ imposed by $\hat{y}$ are readily expressed
as (see~\cite{GoyalVT1998})
\begin{equation}
\label{eq:qfe-M-intervals}
  \left[ \begin{array}{r} F \\ -F \end{array} \right] x
   \leq 
  \left[ \begin{array}{c} \Half \Delta + \hat{y} \\
                          \Half \Delta - \hat{y} \end{array} \right] ,
\end{equation}
where the inequalities are elementwise.
For example, all $2M$ constraints specify a single cell in Fig.~\ref{fig:qfe}(d).
This formulation inspires Algorithm~\ref{alg:qfe-lp},
which is a modification of~\cite[Table I]{GoyalVT1998} using the
principle of maximizing slackness of inequalities that was also
implemented in~\cite{RanganG:01}.
Section~\ref{sec:linearProg} presents analogues to (\ref{eq:qfe-M-intervals})
and Algorithm~\ref{alg:qfe-lp} for FPQ.

\begin{algorithm}
 \caption{Consistent Reconstruction from Scalar-Quantized Frame Expansion}
 \label{alg:qfe-lp}
  \vspace*{3mm}
   \begin{tabular}{rl}
   {\bf Inputs:} & Analysis frame operator $F$, quantization step size $\Delta$,
      and quantized frame expansion $\hat{y}$ \\
   {\bf Output:} & Estimate $\hat{x}$ consistent with $\hat{y}$ and as far from
      the partition boundaries as possible \\[3mm]
     & 1. Let 
          $A = \left[ \begin{array}{rr} F & 1_{M \times 1} \\
                                       -F & 1_{M \times 1} \end{array} \right]$
          and
          $b = \left[ \begin{array}{r} \hat{y} \\
                                     - \hat{y} \end{array} \right]$. \\
     & $\quad$ (Consistency as in (\ref{eq:qfe-M-intervals}) is expressed as
             $A \left[ \begin{array}{c} x \\ 0 \end{array} \right] \leq b$.) \\
     & 2. Let $c = \left[ \begin{array}{c} 0_{N \times 1} \\
                                          -1 \end{array} \right]$. \\
     & 3. Use a linear programming method to minimize
          $c^T \left[ \begin{array}{c} x \\ \delta \end{array} \right]$
          subject to
          $A \left[ \begin{array}{c} x \\ \delta \end{array} \right] \leq b$. \\
     & $\quad$ Return the first $N$ components of the minimizer as $\hat{x}$. \\
   \end{tabular}
\end{algorithm}

The cost of Algorithm~\ref{alg:qfe-lp} may be prohibitive if $M$ is large.
In particular,
Algorithm~\ref{alg:qfe-lp} uses a linear program with $N+1$ variables and
$2M$ constraints, and solving this has cost that is superlinear in $M$.
One way to reduce the cost to linear in $M$ is to use each of $M$
quantized coefficients only once and in a computation with constant cost.
Algorithm~\ref{alg:Rangan-Goyal} uses each constraint 
(\ref{eq:qfe-one-interval}) once, recursively,
in isolation of all other constraints.
It uses $y_k$ by orthogonally projecting a running estimate $\hat{x}_{k-1}$
to the set consistent with (\ref{eq:qfe-one-interval}).
Remarkably, even though the final estimate is generally not consistent
with all $M$ constraints of the form (\ref{eq:qfe-one-interval}),
the optimal $\Theta(M^{-2})$ decay of $\| x - \hat{x} \|^2$
as a function of the number of coefficients $M$ can be
attained under appropriate technical conditions~\cite{RanganG:01,Powell:10}.

\begin{algorithm}
 \caption{Recursive Estimation from Scalar-Quantized Frame Coefficient Sequence}
 \label{alg:Rangan-Goyal}
  \vspace*{3mm}
   \begin{tabular}{rl}
   {\bf Inputs:} & Analysis frame sequence $\{\phi_j\}_{j=1}^M$,
      quantization step size $\Delta$,
      and quantized coefficients $\{\hat{y}_j\}_{j=1}^M$ \\
   {\bf Output:} & Estimate $\hat{x}$ \\[3mm]
     & 1. Let $\hat{x}_0 = 0_{N \times 1}$ and let $k=1$. \\
     & 2. If      $\ip{\hat{x}_{k-1}}{\phi_k} < \hat{y}_k - \half\Delta$,
          let $\hat{x}_k = \hat{x}_{k-1} +
           \Frac{(\hat{y}_k-\half\Delta-\ip{\hat{x}_{k-1}}{\phi_k})\phi_k}
                {\|\phi_k\|^2}$; \\
     & $\quad$ else if $\ip{\hat{x}_{k-1}}{\phi_k} > \hat{y}_k + \half\Delta$,
          let $\hat{x}_k = \hat{x}_{k-1} +
           \Frac{(\hat{y}_k+\half\Delta-\ip{\hat{x}_{k-1}}{\phi_k})\phi_k}
                {\|\phi_k\|^2}$; \\
     & $\quad$ else let $\hat{x}_{k} = \hat{x}_{k-1}$. \\
     & 3. If $k = M$, return $\hat{x}_k$; else increment $k$ and go to step 2. \\
   \end{tabular}
\end{algorithm}

\section{Frame Permutation Quantization}
\label{sec:framePSC}
With background material on permutation source codes and finite frames
in place, we are now prepared to formally introduce frame permutation
quantization.  FPQ is simply PSC applied to a frame expansion.

\subsection{Encoder Definition and Canonical Decoding}
\label{sec:fpq-encoder}
\begin{defi}
A \emph{frame permutation quantizer} with analysis frame $F \in \R^{M \times N}$,
composition $m = (m_1,\,m_2,\ldots,\,m_K)$, and initial codeword $\yinit$ compatible with $m$ encodes $x \in \R^N$
by applying a permutation source code with composition $m$ and initial codeword $\yinit$ to $Fx$.
The \emph{canonical decoding} gives $\hat{x} = F^{\dagger}\hat{y}$,
where $\hat{y}$ is the PSC reconstruction of $y$.
\end{defi}
The two variants of PSCs yield two variants of FPQ\@.
We sometimes use the triple $\(F,m,\yinit\)$ along with a specification
of Variant I or Variant II to refer to such an FPQ\@.

For Variant I, the result of the encoding can be expressed as a permutation $P$
from the permutation matrices of size $M$.
The permutation is such that $PFx$ is $m$-descending.
For uniqueness in the representation $P$ chosen from the set of permutation matrices, we can specify that the first
$m_1$ components of $Py$ are kept in the same order as they appeared in $y$,
the next $m_2$ components of $Py$ are kept in the same order as they
appeared in $y$, etc.
Then $P$ is in a subset $\mathcal{G}(m)$ of the
$M \times M$ permutation matrices and
\begin{equation}
\label{eq:Gm-size}
 | \mathcal{G}(m) | = \frac{M!}{m_1 ! \, m_2 ! \, \cdots \, m_{K} !}.
\end{equation}
Notice that, analogous to the discussion in Section~\ref{sec:psc-partition},
the encoding uses the composition $m$ but not the initial codeword
$\yinit$.
The PSC reconstruction of $y$ is $P^{-1} \yinit$,
so the canonical decoding gives $F^{\dagger} P^{-1} \yinit$.

For Variant II, we will sidestep the differences between the
$\mu_K = 0$ and $\mu_K \neq 0$ cases in Section~\ref{sec:permutationcodes}
by specifying that the signs of the $m_K$ smallest-magnitude components
of $Fx$ are not encoded and $m_K = 0$ is allowed.
Now the result of encoding can be expressed similarly as
$P \in \mathcal{G}(m)$ along with a diagonal matrix $V$ with $\pm 1$
on its diagonal.
These matrices are selected such that the elementwise absolute values
of $VPFx$ are $m$-descending
and also the first $M-m_K$ entries of $VPFx$ are positive.
The last $m_K$ diagonal entries of $V$ do not affect the encoding
and are set to $+1$.
Thus $V$ is in a subset $\mathcal{Q}(m)$ of the
$M \times M$ sign-changing matrices and
\begin{equation}
 | \mathcal{Q}(m) | = 2^{M-m_K}.
\end{equation}
The PSC reconstruction of $y$ is $P^{-1} V^{-1} \yinit$,
so the canonical decoding gives $F^{\dagger} P^{-1} V^{-1} \yinit$.

The sizes of the sets $\mathcal{G}(m)$ and
$\mathcal{G}(m) \times \mathcal{Q}(m)$ are analogous to the
codebook sizes in (\ref{eq:size12}),
and the per-component rates of FPQ are thus defined as
\begin{subequations}
\label{eq:FPQrate}
\begin{equation}
\RI = N^{-1}\log_2 \frac{M!}{m_1 ! \, m_2 ! \, \cdots \, m_{K} !},
\qquad \mbox{for Variant I,}
\end{equation}
and
\begin{equation}
\RII = N^{-1}\left( M - m_K + \log_2 \frac{M!}{m_1 ! \, m_2 ! \, \cdots \, m_{K} !} \right),
\qquad \mbox{for Variant II.}
\end{equation}
\end{subequations}

\subsection{Expressing Consistency Constraints}
\label{sec:fpq-consistency}

Suppose FPQ encoding of $x \in \R^N$ with frame $F \in \R^{M \times N}$,
composition $m = (m_1,m_2,\ldots,m_K)$, and initial codeword $\yinit$ compatible with $m$
results in permutation $P \in \mathcal{G}(m)$
(and, in the case of Variant II, $V \in \mathcal{Q}(m)$)
as described in Section~\ref{sec:fpq-encoder}.
We would like to express constraints on $x$ that are specified by
$(F,m,\yinit,P)$ (or $(F,m,\yinit,P,V)$).
This will provide an explanation of the partitions induced by FPQ
and lead to reconstruction algorithms in Section~\ref{sec:linearProg}.

Knowing that a vector is $m$-descending 
is a specification of many inequalities.
Recall the definitions of the index sets generated by a composition
given in (\ref{eq:index-sets}), and use the same notation with $n_k$s
replaced by $m_k$s.
Then $z$ being $m$-descending implies that for any $i < j$,
$$
  z_k \geq z_\ell
\qquad
  \mbox{for every $k \in \setI_i$ and $\ell \in \setI_j$}.
$$
By transitivity, considering every $i < j$ gives redundant inequalities.
Taking only $j = i+1$, we obtain a full description
\begin{equation}
\label{eq:m-descending-ineq}
  z_k \geq z_\ell
\qquad
  \mbox{for every $k \in \setI_i$ and $\ell \in \setI_{i+1}$ with $i = 1, 2, \ldots, K-1$}.
\end{equation}
For one fixed $(i,\ell)$ pair, (\ref{eq:m-descending-ineq}) gives
$|\setI_i| = m_i$ inequalities, one for each $k \in \setI_i$.
These inequalities can be gathered into an elementwise matrix inequality as
$$
  \left[ \begin{array}{ccc}
          0_{m_i \times M_{i-1}} & I_{m_i} & 0_{m_i \times (M-M_i)} 
         \end{array} \right] z
   \ \geq \
  \left[ \begin{array}{ccc}
          0_{m_i \times (\ell-1)} & 1_{m_i \times 1} & 0_{m_i \times (M-\ell)} 
         \end{array} \right] z
$$
where $M_k = m_1 + m_2 + \cdots + m_k$,
or $\minus{i,\ell}{m} z \geq 0_{m_i \times 1}$ where
\begin{subequations}
\label{eq:minus}
\begin{equation}
\minus{i,\ell}{m} =
  \left[ \begin{array}{ccccc}
          0_{m_i \times M_{i-1}} & I_{m_i} &
          0_{m_i \times (\ell-M_i-1)} & -1_{m_i \times 1} & 0_{m_i \times M-\ell} 
         \end{array} \right]
\end{equation}
is an $m_i \times M$ differencing matrix.
Allowing $\ell$ to vary across $\setI_{i+1}$, we define
the $m_i m_{i+1} \times M$ matrix
\begin{equation}
  \minus{i}{m} = \left[ \begin{array}{c}
        \minus{i,M_i+1}{m} \\[\xtraV]
        \minus{i,M_i+2}{m} \\[\xtraV]
        \vdots             \\[\xtraV]
        \minus{i,M_i+m_{i+1}}{m} 
              \end{array} \right]
\end{equation}
and express all of (\ref{eq:m-descending-ineq}) for one fixed $i$ as
$\minus{i}{m} z \geq 0_{m_i m_{i+1} \times 1}$.

Continuing our recursion, it only remains to gather the inequalities
(\ref{eq:m-descending-ineq}) across $i \in \{1,2,\ldots,K-1\}$.
Let
\begin{equation}
  \minus{}{m} = \left[ \begin{array}{c}
        \minus{1}{m} \\[\xtraV]
        \minus{2}{m} \\[\xtraV]
        \vdots             \\[\xtraV]
        \minus{K-1}{m} 
              \end{array} \right],
\end{equation}
\end{subequations}
which has
\begin{equation}
\label{eq:L-of-m}
L(m) = \sum_{i=1}^{K-1} m_i m_{i+1}
\end{equation}
rows.
The property of $z$ being $m$-descending can be expressed as
$\minus{}{m} z \geq 0_{L(m) \times M}$.
The following example illustrates the form of $\minus{}{m}$:
$$
  \minus{}{(2,3,2)} = \begin{bmatrix}
     1  &   0  &  -1  &   0  &   0  &   0  &   0 \\
     0  &   1  &  -1  &   0  &   0  &   0  &   0 \\
     1  &   0  &   0  &  -1  &   0  &   0  &   0 \\
     0  &   1  &   0  &  -1  &   0  &   0  &   0 \\
     1  &   0  &   0  &   0  &  -1  &   0  &   0 \\
     0  &   1  &   0  &   0  &  -1  &   0  &   0 \\
     0  &   0  &   1  &   0  &   0  &  -1  &   0 \\
     0  &   0  &   0  &   1  &   0  &  -1  &   0 \\
     0  &   0  &   0  &   0  &   1  &  -1  &   0 \\
     0  &   0  &   1  &   0  &   0  &   0  &  -1 \\
     0  &   0  &   0  &   1  &   0  &   0  &  -1 \\
     0  &   0  &   0  &   0  &   1  &   0  &  -1
                      \end{bmatrix}.
$$
Notice the important property that each row of $\minus{}{m}$ has
one $1$ entry and one $-1$ entry with the remaining entries 0\@.
This will be exploited in Section~\ref{sec:consistency}.

Now we can apply these representations to FPQ\@.

\emph{Variant I}:
In this case, we know $PFx$ is $m$-descending.
Consistency is thus simply expressed as
\begin{equation}
\label{eq:fpq-I-consistency}
  \minus{}{m} P F x \ \geq \ 0.
\end{equation}

\emph{Variant II}:  The second variant has an $m$-descending property
after $V$ has made the signs of the significant frame coefficients
(all but last $m_K$) positive:
$\minus{}{m} V P F x \geq 0$.
In addition, we have the nonnegativity of all of the first $M-m_K$
sorted and sign-changed coefficients.  To specify
$$
  \begin{bmatrix} I_{M-m_K} & 0_{(M-m_K) \times M} \end{bmatrix} VPFx
  \ \geq \ 0_{(M-m_K) \times 1}
$$
is redundant with what is expressed with the $m$-descending property.
The added constraints can be applied only to the entries of $VPFx$
with indexes in $\setI_{K-1}$ because all the earlier entries are
already ensured to be larger.
We thus express consistency as
\begin{equation}
\label{eq:fpq-II-consistency}
  \underbrace{\begin{bmatrix}
     & \minus{}{m} & \\
     0_{m_{K-1} \times M_{K-1}} & I_{m_{K-1}} & 0_{m_{K-1} \times m_K}
  \end{bmatrix}}_{\Minus{}{m}}
  V P F x \ \geq \ 0.
\end{equation}

\subsection{Consistent Reconstruction Algorithms}
\label{sec:linearProg}
The constraints (\ref{eq:fpq-I-consistency}) and (\ref{eq:fpq-II-consistency})
both specify unbounded sets, as discussed previously and illustrated
in Fig.~\ref{fig:qfe}(e) and~(f).
To be able to decode FPQs in analogy to Algorithm~\ref{alg:qfe-lp},
we require some additional constraints.  We develop two examples:
a source $x$ bounded to $[-\half,\half]^N$
(e.g., having an i.i.d.\ uniform distribution over $[-\half,\half]$)
or having an i.i.d.\ standard Gaussian distribution.
For the remainder of this section, we consider only Variant I because
adjusting for Variant II using (\ref{eq:fpq-II-consistency}) is easy.

\subsubsection{Source Bounded to $[-\half,\half]^N$}
To impose (\ref{eq:fpq-I-consistency}) along with $x \in [-\half,\half]^N$
is trivial because $x \in [-\half,\half]^N$ is decomposable into $2N$
inequality constraints:
$$
  \left[ \begin{array}{r} I_N \\ -I_N \end{array} \right]  x 
    \ \leq \ \Half 
  \left[ \begin{array}{r} 1_{N \times 1} \\ -1_{N \times 1} \end{array} \right].
$$
A linear programming formulation will return some corner of the consistent
set, depending on the choice of cost function.
The unknown vector $x$ can be augmented with a variable $\delta$
that represents the slackness of the inequality constraint with the
least slack.  Maximizing $\delta$ moves the solution away from the boundary of
the consistent set (partition cell) as much as possible.
Reconstruction using this principle is outlined in Algorithm~\ref{alg:varI-uniform}.

\begin{algorithm}
 \caption{Estimation of Source on $[-\half,\half]^N$ for Variant I Frame Permutation Quantization}
 \label{alg:varI-uniform}
  \vspace*{3mm}
   \begin{tabular}{rl}
   {\bf Inputs:} & Analysis frame operator $F$, composition $m$,
      and FPQ encoding $P$ \\
   {\bf Output:} & Estimate $\hat{x}$ consistent with $(F,m,P)$ and as far from
      the partition boundaries as possible \\[3mm]
     & 1. Let 
          $A = \left[ \begin{array}{rr}
               - \minus{}{m} P F & 1_{L(m) \times 1} \\
                            -I_N & 1_{N \times 1} \\
                             I_N & 1_{N \times 1}
                      \end{array} \right]$
          and
          $b = \Half
               \left[ \begin{array}{r}
                   0_{L(m) \times 1} \\
                   1_{2N \times 1}
                      \end{array} \right]$, \\
     & $\quad$ where $\minus{}{m}$ is defined in 
                (\ref{eq:minus}) and $L(m)$ is defined in (\ref{eq:L-of-m}). \\
     & $\quad$ (Consistency with (\ref{eq:fpq-I-consistency}) and
          $x \in [-\half,\half]^N$ is expressed as
             $A \left[ \begin{array}{c} x \\ 0 \end{array} \right] \leq b$.) \\
     & 2. Let $c = \left[ \begin{array}{c} 0_{N \times 1} \\
                                          -1 \end{array} \right]$. \\
     & 3. Use a linear programming method to minimize
          $c^T \left[ \begin{array}{c} x \\ \delta \end{array} \right]$
          subject to
          $A \left[ \begin{array}{c} x \\ \delta \end{array} \right] \leq b$. \\
     & $\quad$ Return the first $N$ components of the minimizer as $\hat{x}$. \\
   \end{tabular}
\end{algorithm}

If the source $x$ is random and the distribution $p(x)$ is known,
then one could optimize some criterion explicitly.
For example, one could maximize $p(x)$ over the consistent set
or compute the centroid of the consistent set with respect to $p(x)$.
This would improve upon reconstructions computed with
Algorithm~\ref{alg:varI-uniform} but presumably increase complexity greatly.

\subsubsection{Source with i.i.d.\ Standard Gaussian Distribution}
Suppose $x$ has i.i.d.\ Gaussian
components with mean zero and unit variance.
Since the source support is unbounded, something beyond consistency
must be used in reconstruction. 
Here we use a quadratic program to find a good bounded, consistent
estimate and combine this with the average value of $\|x\|$.

The problem with using
(\ref{eq:fpq-I-consistency}) combined with maximization of minimum
slackness alone (without any additional boundedness constraints)
is that for any purported solution, multiplying by a scalar larger
than 1 will increase the slackness of all the constraints.
Thus, any solution technique will naturally and correctly have
$\| \hat{x} \| \rightarrow \infty$.
Actually, because the partition cells are convex cones, we should
not hope to recover the radial component of $x$ from the partition.
Instead, we should only hope to recover a good estimate of $x / \| x \|$.

To estimate the angular component $x / \| x \|$ from the partition,
it would be convenient to maximize minimum slackness
while also imposing a constraint of $\| \hat{x} \| = 1$.
Unfortunately, this is a nonconvex constraint.
It can be replaced by $\| \hat{x} \| \leq 1$ because slackness is
proportional to $\| \hat{x} \|$.
This suggests the optimization
$$
  \mbox{maximize $\delta$ subject to $\| x \| \leq 1$ and
  $\minus{}{m} P F x \ \geq \ \delta 1_{L(m) \times 1}$}.
$$
Denoting the $x$ at the optimum by $\hat{x}_{\rm ang}$,
we still need to choose the radial component, or length, of $\hat{x}$.

For the $\mathcal{N}(0,I_N)$ source, the mean length is~\cite{Sakrison:68}
$$
  E[\| x \|] = \frac{\sqrt{2\pi}}{\beta(N/2,1/2)} \approx \sqrt{N - 1/2}.
$$
We can combine this with $\hat{x}_{\rm ang}$ to obtain a reconstruction
$\hat{x}$.

We use a slightly different formulation to have a quadratic program
in standard form.
We combine the radial component constraint with the goal of maximizing
slackness to obtain
$$
  \mbox{minimize $\half x^T x - \lambda \delta$ subject to 
  $- \minus{}{m} P F x \ \leq \ -\delta 1_{L(m) \times 1}$},
$$
where $\lambda$ trades off slackness against the radial component of $x$.
Since the radial component will be replaced with its expectation,
the choice of $\lambda$ is immaterial;
it is set to 1 in Algorithm~\ref{alg:varI-Gaussian}.

\begin{algorithm}
 \caption{Estimation of $\mathcal{N}(0,I_N)$ Source for Variant I Frame Permutation Quantization}
 \label{alg:varI-Gaussian}
  \vspace*{3mm}
   \begin{tabular}{rl}
   {\bf Inputs:} & Analysis frame operator $F$, composition $m$,
      and FPQ encoding $P$ \\
   {\bf Output:} & Estimate $\hat{x}$ consistent with $(F,m,P)$ and as far from
      the partition boundaries as possible \\
                 & while keeping $\|\hat{x}\| = E[\|x\|]$ \\[3mm]
     & 1. Let 
          $A = \left[ \begin{array}{rr}
               - \minus{}{m} P F & 1_{L(m) \times 1}
                      \end{array} \right]$
          and
          $b = 0_{L(m) \times 1}$, \\
     & $\quad$ where $\minus{}{m}$ is defined in 
                (\ref{eq:minus}) and $L(m)$ is defined in (\ref{eq:L-of-m}). \\
     & $\quad$ (Consistency with (\ref{eq:fpq-I-consistency})
          is expressed as
             $A \left[ \begin{array}{c} x \\ 0 \end{array} \right] \leq b$.) \\
     & 2. Let $c = \left[ \begin{array}{c} 0_{N \times 1} \\
                                          -1 \end{array} \right]$
          and $H = \begin{bmatrix}
                      I_N & 0_{N \times 1} \\
                  0_{1 \times N} & 0
                   \end{bmatrix}$. \\
     & 3. Use a quadratic programming method to minimize
          $\half \left[ \begin{array}{c} x \\ \delta \end{array} \right]^T
             H \left[ \begin{array}{c} x \\ \delta \end{array} \right] +
           c^T \left[ \begin{array}{c} x \\ \delta \end{array} \right]$ \\
     & $\quad$     subject to
          $A \left[ \begin{array}{c} x \\ \delta \end{array} \right] \leq b$. 
     Denote the first $N$ components of the minimizer as $\hat{x}_{\rm ang}$. \\
     & 4. Return $(\sqrt{2\pi}/\beta(N/2,1/2)) \, \hat{x}_{\rm ang}$.
   \end{tabular}
\end{algorithm}

\subsection{Recursive Estimation Algorithms}
\label{sec:recursive}

Algorithms~\ref{alg:varI-uniform} and~\ref{alg:varI-Gaussian}
are practical for small values of $N$ and $M$,
as one would encounter in data compression applications,
but not for large values of $N$ and $M$,
as may be of interest in data acquisition.
Specifically, Algorithm~\ref{alg:varI-uniform} uses a linear program with
$N+1$ variables and $L(m)+2N$ constraints,
while Algorithm~\ref{alg:varI-Gaussian} uses a quadratic program with
$N+1$ variables and $L(m)$ constraints.
The costs of these computations are superlinear in $L(m)$,
and $L(m)$ is at least $M-1$.

One way to lower the reconstruction complexity is to sacrifice global
consistency in favor of recursive computability,
in analogy to Algorithm~\ref{alg:Rangan-Goyal}.
For recursive algorithms, we restrict our attention to the composition
$m = (1,\,1,\,\ldots,\,1)$.  We also again restrict our attention to
Variant I because adjusting for Variant II is straightforward.

With all-1s compositions, FPQ encoding produces a
\emph{successive} or \emph{embedded} representation of $x$:
a representation with a $k$-element analysis frame is a ranking of
$\{\ip{x}{\phi_j}\}_{j=1}^k$ (for Variant I),
and adding a vector $\phi_{k+1}$ to the analysis frame
amounts to inserting $\ip{x}{\phi_{k+1}}$ in the ranked list.
Equivalently, the encoding of $x$ with a $k$-element frame is the set
$$
  \sign(\ip{x}{\phi_i}-\ip{x}{\phi_j}),
\qquad
  i,j \in \{1,\,2,\,\ldots,\,k\},
$$
and adding $\phi_{k+1}$ to the analysis frame adds
$$
  \sign(\ip{x}{\phi_{k+1}}-\ip{x}{\phi_j}),
\qquad
  j \in \{1,\,2,\,\ldots,\,k\},
$$
to the representation without removing any of the previous information.

In estimating $x$ from FPQ with the full $(k+1)$-element analysis frame,
one could impose
\begin{equation}
\label{eq:nonrecursive-constraint}
  \sign(\ip{\hat{x}}{\phi_i-\phi_j}) = \sign(\ip{x}{\phi_i-\phi_j}),
\qquad
  i,j \in \{1,\,2,\,\ldots,\,k+1\}
\end{equation}
(equivalent to (\ref{eq:fpq-I-consistency}))
for \emph{global} consistency.
However, for a recursive computation we will compute an estimate $\hat{x}_{k+1}$
from an estimate $\hat{x}_k$ and some subset of the constraints
\begin{equation}
\label{eq:recursive-constraint}
  \sign(\ip{\hat{x}}{\phi_{k+1}-\phi_j}) = \sign(\ip{x}{\phi_{k+1}-\phi_j}),
\qquad
  j \in \{1,\,2,\,\ldots,\,k\}.
\end{equation}
Updating $\hat{x}_k$ to satisfy any of the constraints
(\ref{eq:recursive-constraint}) can cause constraints
(\ref{eq:nonrecursive-constraint}) with $i < k+1$ to be violated,
so a strategy of imposing \emph{local} consistency does not ensure
global consistency.
However, we will demonstrate by extending results
from~\cite{RanganG:01,Powell:10} that optimal MSE decay as a function
of $M$ can still be obtained.

A recursive computation that uses local consistency is described
explicitly in Algorithm~\ref{alg:varI-recursive}.
For each $k$, the set $\setJ_k$ represents the indexes $j$ for
which the constraint (\ref{eq:recursive-constraint}) is used.
Any one of them is used (in Step 3b)
by orthogonally projecting the running estimate
$\hat{x}_{k-1}$ to the half-space consistent with
(\ref{eq:recursive-constraint}).
This yields a monotonicity result analogous to~\cite[Thm.~1]{RanganG:01}
and~\cite[Lem.~7.2]{Powell:10}:

\begin{algorithm}
 \caption{Recursive Estimation for Variant I Frame Permutation Quantization}
 \label{alg:varI-recursive}
  \vspace*{3mm}
   \begin{tabular}{rl}
   {\bf Inputs:} & Analysis frame sequence $\{\phi_j\}_{j=1}^M$,
                   FPQ representation $\nu_{i,j} = \sign(\ip{x}{\phi_i-\phi_j})$
                   for $j=1,2,\ldots,M$, \\ & $i=1,2,\ldots,j-1$, and
                   index sets $\setJ_2,\,\setJ_3,\,\ldots,\,\setJ_M$ with
                   $\setJ_k \subset \{1,\,2,\,\ldots,\,k-1\}$ \\
   {\bf Output:} & Estimate $\hat{x}_M$ \\[3mm]
     & 1. Let $\hat{x}_1$ be a vector chosen uniformly at random from the
          unit sphere in $\R^N$ and let $k=2$. \\
     & 2. Let $\hat{x}_k = \hat{x}_{k-1}$. \\
     & 3. For each $j \in \setJ_k$, taken in random order: \\
     & $\quad$ 3a. Let $\psi_{k,j} = \phi_k - \phi_j$. \\
     & $\quad$ 3b. If $\sign(\ip{\hat{x}_{k}}{\psi_{k,j}}) \neq \nu_{k,j}$,
          update $\hat{x}_k$ to
          $\hat{x}_{k} - \Frac{\ip{\hat{x}_{k}}{\psi_{k,j}}\psi_{k,j}}{\|\psi_{k,j}\|^2}$. \\
     & 4. If $k = M$, return $\hat{x}_k/\|\hat{x}_k\|$;
          else increment $k$ and go to step 2. \\
   \end{tabular}
\end{algorithm}

\begin{thm}\label{thm:recursive-monotonicity}
Let $x \in \R^N$ be a unit vector.  The sequence of estimates produced
in Algorithm~\ref{alg:varI-recursive} satisfies
$$
  \| x - \hat{x}_{k+1} \| \leq \| x - \hat{x}_k \|.
$$
\end{thm}
\begin{proof}
  Since Step 3b is an orthogonal projection to a convex set containing $x$,
  no occurrence of this step can increase the estimation error.
\end{proof}

The number of projection steps in Algorithm~\ref{alg:varI-recursive}
depends on the sizes of the $\setJ_k$s.
At one extreme, each $\setJ_k$ is a singleton so there are $M-1$ projections.
At the other extreme, each $\setJ_k$ has $k-1$ elements and there are
$\half M (M-1)$ projections.
The empirical behavior in Section~\ref{sec:numerical-recursive}
shows that the MSE decays as the square of the number of projection steps.
This behavior is provable in some cases.  One such result is
the following theorem, analogous to~\cite[Thm.~2]{RanganG:01}
and presumably extendable to match~\cite[Thm.~7.9]{Powell:10}:

\begin{thm}\label{thm:recursive-rate}
Let $x$ and $\{\phi_k\}_{k=1}^\infty$ be i.i.d.\ vectors drawn from the
uniform distribution on the unit sphere in $\R^N$, and
for each $k \in \{2,\,3,\,\ldots\}$, let $\setJ_k = \{1\}$.
Then the normalized sequence of estimates produced by
Algorithm~\ref{alg:varI-recursive} satisfies
\begin{equation}
 \label{eq:recursive-decay}
  \| x - \hat{x}_k \|^2 \, k^p \rightarrow 0
\quad
\mbox{almost surely, for every $p < 2$}.
\end{equation}
\end{thm}
\begin{proof}
We give only a brief sketch of a proof since the main ideas have been developed
by Rangan and Goyal~\cite{RanganG:01} and Powell~\cite{Powell:10}.

According to~\cite[Thm.~2]{RanganG:01},
Algorithm~\ref{alg:Rangan-Goyal} gives performance satisfying
(\ref{eq:recursive-decay}) under the following assumptions:
\begin{enumerate}
\item Quantization noise is on a known, bounded interval;
\item the frame sequence is i.i.d.\ and independent of the quantization noise;
      and
\item the frame vectors are bounded and satisfy
      $$E[|\ip{z}{\phi_k}|] \geq \varepsilon \| z \| \quad
      \mbox{for all $z \in \R^N$}
      $$
      for some $\varepsilon > 0$.
\end{enumerate}
Assumption 1 can be loosened to quantization noise known to lie in
$[-1,\infty)$ or $(-\infty,1]$ without changing the conclusion that
(\ref{eq:recursive-decay}) holds;
it is qualitatively like discarding half of the quantized frame coefficients
since each frame coefficient gives one half-space constraint rather than two.
Assumption 3 is a very mild condition that simply ensures that there is
no nonzero vector $z \in \R^N$ such that all of the probability mass of the
frame vector distribution is orthogonal to $z$; such orthogonality implies
$\{\ip{x}{\phi_k}\}_{k=1}^\infty$ gives no information about the component
of $x$ in the subspace generated by $z$ and hence makes recovery
from the inner product sequence impossible even without quantization.

An FPQ representation is through $\sign(\ip{x}{\psi_{k,j}})$ where
$\psi_{k,j} = \phi_k - \phi_j$, $j \in \setJ_k$, and $k = 2,\,3,\,\ldots$.
This is equivalent to a quantized
frame expansion with analysis vectors $\psi_{k,j}$
and quantization noise bounded to $[-1,\infty)$
when the signum function returns $1$ and
bounded to $(-\infty,1]$ when the signum function returns $-1$.
Due to the symmetric distribution of $x$ and its independence of
the frame sequence,
the quantization noise is independent of the frame sequence.

We are considering reconstruction where each $\setJ_k$ is $\{1\}$,
so we are interested in whether quantized versions of
$\{\ip{x}{\psi_{k,1}}\}$ are adequate to ensure the error decay
(\ref{eq:recursive-decay}).
While the vectors $\{\psi_{k,1}\}_{k=2}^\infty$ are identically distributed
and (through spherical symmetry) satisfy Assumption~3 above,
they are not independent.
Nevertheless, we can employ~\cite[Thm.~2]{RanganG:01} as follows:
Conditioned on any value of $\phi_1$, the vectors
$\{\psi_{k,1}\}_{k=2}^\infty$ are conditionally independent, so 
by eliminating the radial component of $x$,
one can conclude that the error decay (\ref{eq:recursive-decay}) holds
conditionally.
Since the almost sure convergence in (\ref{eq:recursive-decay}) holds
under every conditional probability law specified by $\phi_1$,
it must also hold unconditionally.
\end{proof}

Extensions of~\cite[Thm.~2]{RanganG:01} and~\cite[Thm.~7.9]{Powell:10}
to non-i.i.d.\ frame sequences would lead to extensions of
Theorem~\ref{thm:recursive-rate} beyond singleton $\setJ_k$s.

\section{Conditions on the Choice of Frame}
\label{sec:consistency}
In this section, we provide necessary and sufficient conditions so that a
linear reconstruction is also consistent. We first consider a general linear
reconstruction, $\hat{x}=R\hat{y}$, where $R$ is some $N \times M$ matrix
and $\hat{y}$ is a decoding of the PSC of $y$.
We then restrict attention to canonical reconstruction, where $R=F^{\dagger}$.
For each case, we describe all possible choices of a ``good'' frame $F$,
in the sense of the consistency of the linear reconstruction.

\subsection{Arbitrary Linear Reconstruction}
We begin by introducing a useful term.

\begin{defi}
A matrix is called \emph{column-constant} when each column of the
matrix is a constant.  The set of all $M \times M$ column-constant
matrices is denoted $\mathcal{J}$.
\end{defi}

We now give our main results for arbitrary linear reconstruction
combined with FPQ decoding of an estimate of $y$.

\begin{thm}\label{thm:sufficient}
Suppose $A = FR = \c\id_{M}+J$ for some $\c\geq 0$ and $J \in\mathcal{J}$.
Then the linear reconstruction $\hat{x}=R\hat{y}$ is consistent with
Variant~I FPQ encoding using frame $F$, an arbitrary composition and an arbitrary Variant~I initial codeword compatible with it.
\end{thm}
\begin{proof}
We start the proof by pointing out two special properties of
any matrix $J \in \mathcal{J}$:
\begin{eqnarray*}
  \mbox{(P1)}&& P J P^{-1} \in\mathcal{J} \qquad \mbox{for any permutation matrix $P$; and} \\
  \mbox{(P2)}&& \minus{}{m} J = 0_{L(m) \times 1}
     \qquad \mbox{for any composition $m$.}
\end{eqnarray*}
(P1) follows from the fact that neither left multiplying by $P$
  nor right multiplying by $P^{-1}$ disturbs column-constancy.
(P2) is true because each row of $\minus{}{m}$ has zero entries except for
one $1$ and one $-1$.

Suppose $m=\(m_1,m_2,\ldots,m_K\)$ is an arbitrary composition of $M$ and $\yinit$ is an arbitrary Variant~I initial codeword compatible with $m$.
Let $P$ be the Variant I FPQ encoding of $x$ using $\(F,m,\yinit\)$.
We would like to check that $\hat{x} = R\hat{y}$ is consistent with the encoding $P$.
This is verified through the following computation:
\begin{eqnarray}
\minus{}{m} P F \hat{x}
   &=& \minus{}{m} P F R \hat{y} \nonumber\\
   &=& \minus{}{m} P F R P^{-1} \yinit \label{eq:suffI-1} \\
   &=& \minus{}{m} P A P^{-1} \yinit \nonumber \\
   &=& \minus{}{m} P \left(\c\id_{M}+J\right) P^{-1} \yinit \label{eq:hypo} \\
   &=& \c\,\minus{}{m}\,\yinit + \minus{}{m}\hat{J}\yinit \qquad
         \mbox{for some $\hat{J} \in \mathcal{J}$} \label{eq:p1}\\
   &=& \c\,\minus{}{m}\,\yinit \label{eq:p2}\\
   &\geq & 0_{L(m) \times 1}, \label{eq:positive2}
\end{eqnarray}
where (\ref{eq:suffI-1}) uses the conventional decoding of a PSC;
(\ref{eq:hypo}) follows from the hypothesis of the theorem on $A$;
(\ref{eq:p1}) follows from (P1);
(\ref{eq:p2}) follows from (P2);
and (\ref{eq:positive2}) follows from the definition of Variant I initial codewords compatible with $m$,
and the nonnegativity of $\c$.  This completes the proof.
\end{proof}

The key point of the proof of Theorem~\ref{thm:sufficient} is showing that the
inequality
\begin{equation}
\minus{}{m} P A P^{-1} \yinit \geq 0 \label{eq:keyIneq},
\end{equation}
where $A=FR$, holds for every composition $m$ and every initial codeword $\yinit$ compatible with it.
It turns out that the form of matrix $A$ given in Theorem~\ref{thm:sufficient}
is the unique form that guarantees that (\ref{eq:keyIneq}) holds for every pair $(m,\yinit)$.
In other words, the condition on $A$ that is sufficient for any composition $m$ and any initial codeword $\yinit$ compatible with it is also necessary for consistency for every pair $(m,\yinit)$.

\begin{thm}\label{thm:necessary}
Consider Variant I FPQ using frame $F$ with $M\geq 3$.
If linear reconstruction $\hat{x}=R\hat{y}$ is consistent with every composition and every Variant I initial codeword compatible with it, then matrix $A=FR$ must be of the form
$\c\id_M + J$, where $\c\geq 0$ and $J \in \mathcal{J}$.
\end{thm}
\begin{proof}
See Section~\ref{app:necessary}.
\end{proof}

The column-constant matrices are those obtained by multiplying an
all-1s matrix on the right by a diagonal matrix.
Thus, except in the case of an all-0s matrix, a column-constant
matrix has rank 1\@.
A matrix of the form $aI_M + J$ where $a > 0$ and $J \in \mathcal{J}$
thus has rank $M-1$ or $M$.
Since $A = FR$ has rank at most $N$ because of the dimensions of $F$ and $R$,
the necessary and sufficient condition from Theorems~\ref{thm:sufficient}
and~\ref{thm:necessary} imply $M = N$ or $M = N+1$.

Similar necessary and sufficient conditions can be derived for
linear reconstruction of Variant II FPQs. Since the partition cell
associated with a codeword of a Variant II FPQ is much smaller than
that of the corresponding Variant I FPQ,
we expect the condition for a linear reconstruction to be consistent
to be more restrictive than that given in Theorems~\ref{thm:sufficient}
and~\ref{thm:necessary}.
The following two theorems show that this is indeed the case.

\begin{thm}\label{thm:varIIsuf}
Suppose $A = FR = \c\id_{M}$ for some $\c\geq 0$ and $M = N$.
Then the linear reconstruction $\hat{x}=R\hat{y}$ is consistent with
Variant~II FPQ encoding using frame $F$, an arbitrary composition, and an arbitrary Variant II initial codeword compatible with it.
\end{thm}
\begin{proof}
Suppose that $m =(m_1,m_2,\ldots,m_K)$ is an arbitrary composition of $M$ and $\yinit$ is an arbitrary Variant II initial codeword compatible with it. Let $(P,V)$ be the Variant II FPQ encoding of $x$ using $\(F,m,\yinit\)$.
We would like to check that $\hat{x} = R\hat{y}$ is consistent with the encoding $(P,V)$.
This is verified through the following computation:
\begin{eqnarray}
\Minus{}{m} V P F \hat{x}
   &=& \Minus{}{m} V P F R     \hat{y} \nonumber\\
   &=& \Minus{}{m} V P F R     P^{-1} V^{-1} \yinit \label{eq:suffII-1} \\
   &=& \Minus{}{m} V P A       P^{-1} V^{-1} \yinit \nonumber \\
   &=& \Minus{}{m} V P \c\id_M P^{-1} V^{-1} \yinit \label{eq:hypoII} \\
   &=& \c\,\Minus{}{m}\,\yinit \nonumber \\
   &\geq & 0_{L(m) \times 1}, \label{eq:positive2-II}
\end{eqnarray}
where (\ref{eq:suffII-1}) uses the conventional decoding of a PSC;
(\ref{eq:hypoII}) follows from the hypothesis of the theorem on $A$;
and (\ref{eq:positive2-II}) follows from the definition of Variant II
initial codewords compatible with $m$, and the nonnegativity of $\c$.
This completes the proof.
\end{proof}

\begin{thm}\label{thm:varIInec}
Consider Variant II FPQ using frame $F$ with $M\geq 3$.
If linear reconstruction $\hat{x}=R\hat{y}$ is consistent with every composition and every Variant II initial codeword compatible with it, then matrix $A=FR$ must be of the form
$\c\id_M$, where $\c\geq 0$ and $M = N$.
\end{thm}
\begin{proof}
See Section~\ref{app:VarIInec}.
\end{proof}

The two theorems above show that,
if we insist on linear consistent reconstructions for Variant II FPQs,
the frame must degenerate into a basis.
For nonlinear consistent reconstructions,
we could use algorithms analogous to those presented in
Section~\ref{sec:linearProg} for an arbitrary frame that is not necessarily a basis.

\subsection{Canonical Reconstruction}
We now restrict the linear reconstruction to use the canonical dual;
i.e., $R$ is restricted to be the pseudo-inverse $F^{\dagger}=\pseudo$.
The following corollary characterizes the non-trivial frames for
which canonical reconstructions are consistent.

\begin{coro}\label{coro:pseudoInverse}
Consider Variant I FPQ using rank-$N$ frame $F$ with $M > N$ and $M\geq 3$.
For canonical reconstruction to be consistent with every composition and every Variant I initial codeword compatible with it, it is necessary and sufficient to have
$M = N+1$ and $A = FF^{\dagger} = \id_M-\frac{1}{M}J_M$,
where $J_M$ is the $M \times M$ all-1s matrix.
\end{coro}
\begin{proof}
Sufficiency follows immediately from Theorem~\ref{thm:sufficient}.
From Theorem~\ref{thm:necessary}, it is necessary to have
$A=FF^{\dagger} = \c\id_M+J$ for some $\c \geq 0$ and $J \in \mathcal{J}$.
The rank condition further implies $a > 0$, so we must have $M = N+1$
by the argument following Theorem~\ref{thm:necessary}.
Now since $A$ is an orthogonal projection operator,
it is self-adjoint so
\begin{equation}
\c\id_M+J = \left(\c\id_M+J\right)^{*}=\c\id_M+J^{*}.
\end{equation}
Thus, $J=J^{*}$, and it follows that $J=b J_M$, for some constant $b$.
The idempotence of $A$ gives
\begin{eqnarray}
\c\id_M+bJ_M&=&\left(\c\id_M+bJ_M\right)^{2}\nonumber\\
&=& \c^2\id_M+(2\c b+b^2M)J_M.\label{eq:idem}
\end{eqnarray}
Equating the two sides of (\ref{eq:idem}) yields
$\c  =  1$ and $b  = -\Frac{1}{M}$
as desired.
\end{proof}

We continue to add more constraints to frame $F$.
Tightness and equal norm are common requirements in
frame design~\cite{KovacevicC:07a}.
By imposing tightness and unit norm on our analysis frame,
we can progress a bit further from Corollary~\ref{coro:pseudoInverse}
to derive the form of $FF^{*}.$
\begin{coro}\label{coro:UNTF}
Consider Variant I FPQ using unit-norm tight frame $F$ with $M > N$ and $M\geq 3$.
For canonical reconstruction to be consistent for every composition and every Variant~I initial codeword compatible with it, it is necessary and sufficient to have
$M = N+1$ and
\begin{equation}\label{eq:equiangular}
FF^{*}=
\begin{bmatrix}
1&-\frac{1}{N}&\cdots&-\frac{1}{N}\\
-\frac{1}{N}&1&\cdots&-\frac{1}{N}\\
\vdots&\vdots&\ddots&\vdots\\
-\frac{1}{N}&-\frac{1}{N}&\cdots&1
\end{bmatrix}.
\end{equation}
\end{coro}
\begin{proof}
Corollary~\ref{coro:pseudoInverse} asserts that $M=N+1$ and
\begin{equation}\label{eq:A}
F(F^{*}F)^{-1}F^{*}=
\begin{bmatrix}
\frac{N}{M}&-\frac{1}{M}&\cdots&-\frac{1}{M}\\
-\frac{1}{M}&\frac{N}{M}&\cdots&-\frac{1}{M}\\
\vdots&\vdots&\ddots&\vdots\\
-\frac{1}{M}&-\frac{1}{M}&\cdots&\frac{N}{M}
\end{bmatrix}.
\end{equation}
On the other hand, the tightness of a unit-norm frame $F$ implies
\begin{equation}\label{eq:tightness}
(F^{*}F)^{-1}=\frac{N}{M}\id_N.
\end{equation}
Combining (\ref{eq:A}) with (\ref{eq:tightness}), we get (\ref{eq:equiangular}).
\end{proof}

Recall that a UNTF that satisfies (\ref{eq:equiangular}) is a restricted ETF\@.
Therefore Corollary~\ref{coro:UNTF} together with Proposition~\ref{prop:ETF}
gives us a complete characterization of UNTFs that are ``good'' in the
sense of canonical reconstruction being consistent.
\begin{coro}
Consider Variant I FPQ using unit-norm tight frame $F$ with $M > N$ and $M\geq 3$.
For canonical reconstruction to be consistent for every composition and every Variant~I initial codeword compatible with it, it is necessary and sufficient for $F$ to be a modulated HTF or a Type~I or Type~II equivalent.
\end{coro}

\section{Numerical Results}\label{sec:simulations}
In this section, we provide simulations to demonstrate some properties of FPQ\@.
For data compression, we demonstrate that FPQ with decoding
using Algorithms~\ref{alg:varI-uniform} and~\ref{alg:varI-Gaussian}
can give performance better than entropy-constrained scalar quantization (ECSQ)
and ordinary PSC for certain combinations of signal dimension and rate.
For data acquisition, we demonstrate that FPQ with recursive estimation
through Algorithm~\ref{alg:varI-recursive} empirically gives the
optimal decay of MSE,
inversely proportional to the square of the number of orthogonal
projection steps, validating Theorem~\ref{thm:recursive-rate}
but also suggesting that this holds more generally.

\subsection{Fixed-Rate Compression Experiments}
All FPQ compression simulations use modulated harmonic tight frames and
are based on implementations of
Algorithms~\ref{alg:varI-uniform} and~\ref{alg:varI-Gaussian}
using MATLAB, with linear programming and quadratic programming
provided by the Optimization Toolbox.
For every data point shown, the distortion represents a sample mean
estimate of $N^{-1} E[ \| x - \hat{x} \|^2]$ over at least $10^6$ trials.
Testing was done with exhaustive enumeration of the relevant
compositions.  This makes the complexity of simulation high,
and thus experiments are only shown for small $N$ and $M$.
Recall the encoding complexity of FPQ is low, $O(M \log M)$ operations.
The decoding complexity is polynomial in $M$ for either of the algorithms
presented explicitly, and in some applications it could be worthwhile
to precompute the entire codebook at the decoder.
Thus much larger values of $N$ and $M$ than used here
may be practical.

\emph{Uniform source.}
Let $x$ have i.i.d.\ components uniformly distributed on $[-\half,\half]$.
Algorithm~\ref{alg:varI-uniform} is clearly well-suited to this source
since the support of $x$ is properly specified and reconstructions near
the centers of cells is nearly optimal.
Fig.~\ref{fig:sim-uniform} summarizes the performance of Variant I FPQ
for several frames and an enormous number of compositions.
Also shown are the performances of ordinary PSC and
optimal ECSQ~\cite{GoyalSW:01,GyorgiL:00}\@.
Approaching the indicated performance of ECSQ requires entropy coding over
many quantized symbols.

\begin{figure}
 \begin{center}
  \begin{tabular}{cc}
   \includegraphics[width=\widthA]{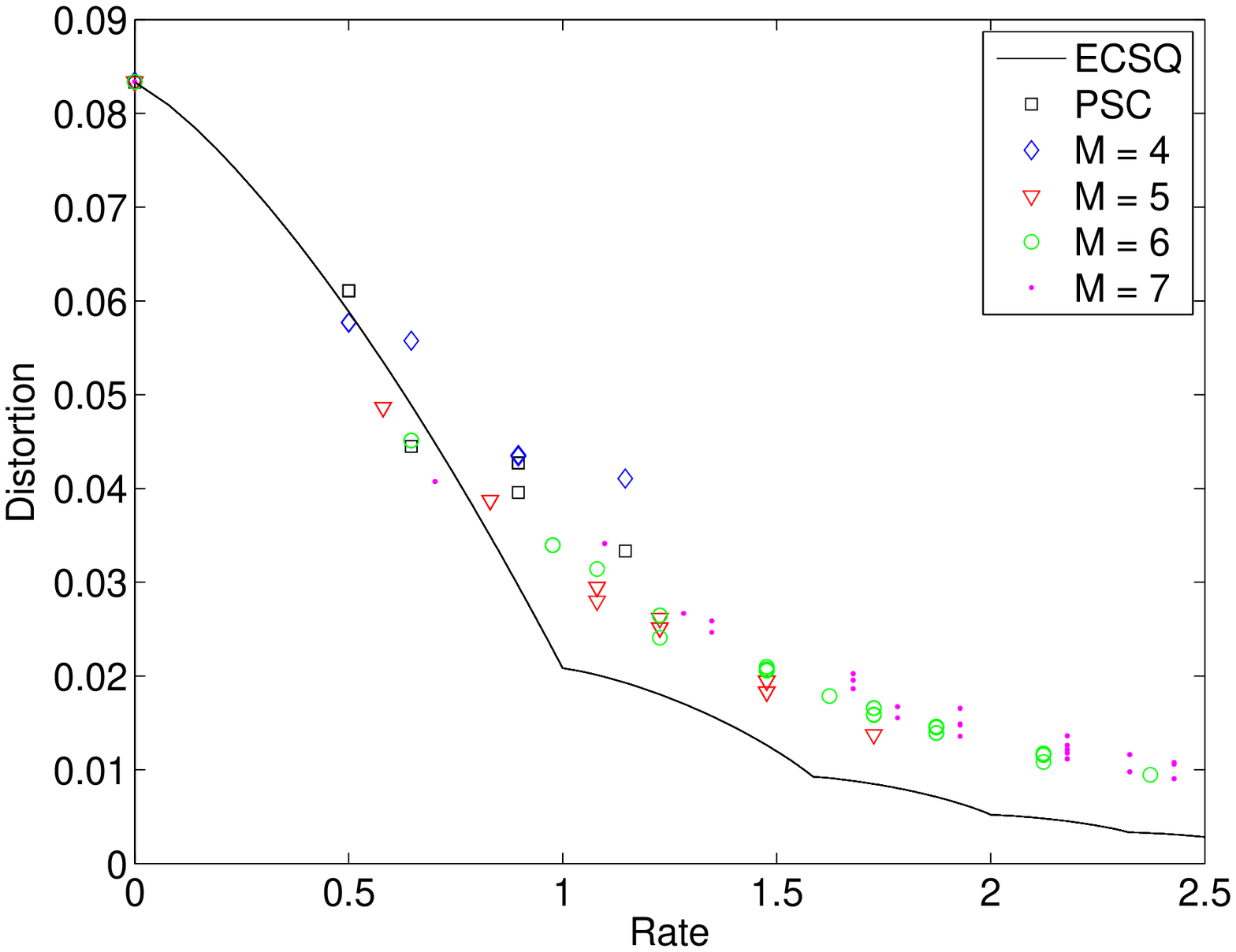} &
   \includegraphics[width=\widthA]{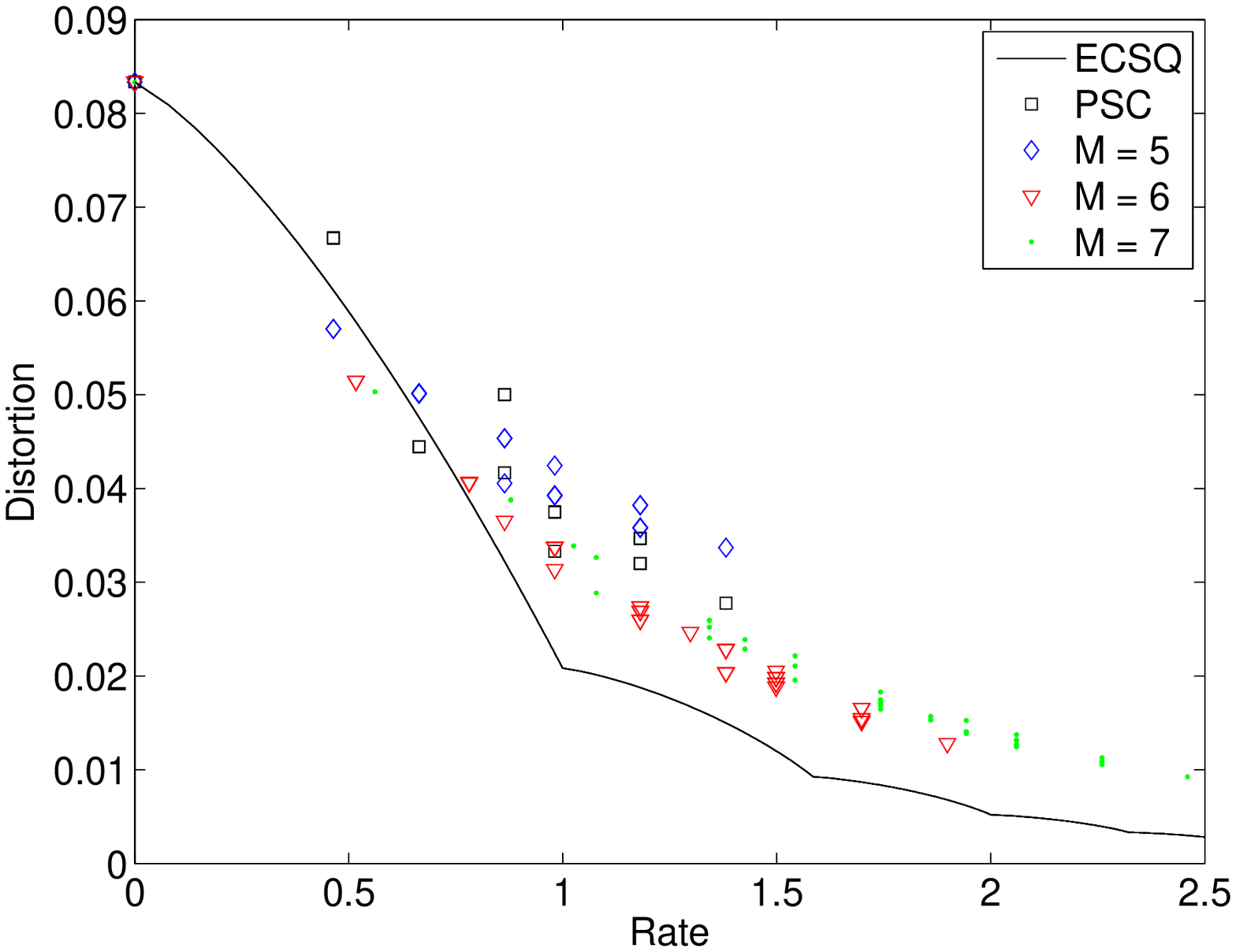} \\[1mm]
   {\small (a) $N=4$} &
   {\small (b) $N=5$}
  \end{tabular}
 \end{center}
 \caption{Performance of Variant I FPQ on an i.i.d.\ uniform($[-\half,\half]$)
   source using modulated harmonic tight frames ranging in size from $N$ to 7\@.
   Also shown are the performances of ordinary PSC
   (equivalent to FPQ with frame $F = I_N$),
   and optimal entropy-constrained scalar quantization.}
 \label{fig:sim-uniform}
\end{figure}

Using $F = I_N$ makes FPQ reduce to ordinary PSC\@.
We see that, consistent with results in~\cite{GoyalSW:01},
PSC is sometimes better than ECSQ\@.
Next notice that FPQ is not identical to PSC when $F$ is square but
not the identity matrix.  The modulated harmonic frame with $M = N$
provides an orthogonal matrix $F$.  The set of rates obtained with
$M = N$ is the same as PSC, but since the source is not rotationally-invariant, 
the partitions and hence distortions are not the same;
the distortion is sometimes better and sometime worse.
Increasing $M$ gives more operating points---some of which are better
than those for lower $M$---and
a higher maximum rate.\footnote{A discussion of the density of PSC rates is
given in~\cite[App.~B]{NguyenVG:10}.}
In particular, for both $N=4$ and $N=5$, it seems that $M=N+1$ gives
several operating points better than those obtainable with larger or
smaller values of $M$.

\emph{Gaussian source.}
Let $x$ have the $\mathcal{N}(0,I_N)$ distribution.
Algorithm~\ref{alg:varI-Gaussian} is designed precisely for this source.
Fig.~\ref{fig:sim-Gaussian} summarizes the performance of Variant I FPQ
with decoding using Algorithm~\ref{alg:varI-Gaussian}.
Also shown are the distortion--rate bound and the performances of
two types of entropy-constrained scalar quantization:
uniform thresholds with uniform codewords (labeled ECUSQ)
and uniform thresholds with optimal codewords (labeled ECSQ)\@.
At all rates, the latter is a very close approximation to optimal ECSQ;
in particular,
it has optimal rate--distortion slope at rate zero~\cite{MarcoN2006}.
Of course, the distortion--rate bound
can only be approached with $N \rightarrow \infty$; it is not presented
as a competitive alternative to FPQ for $N=4$ and $N=5$.

\begin{figure}
 \begin{center}
  \begin{tabular}{cc}
   \includegraphics[width=\widthA]{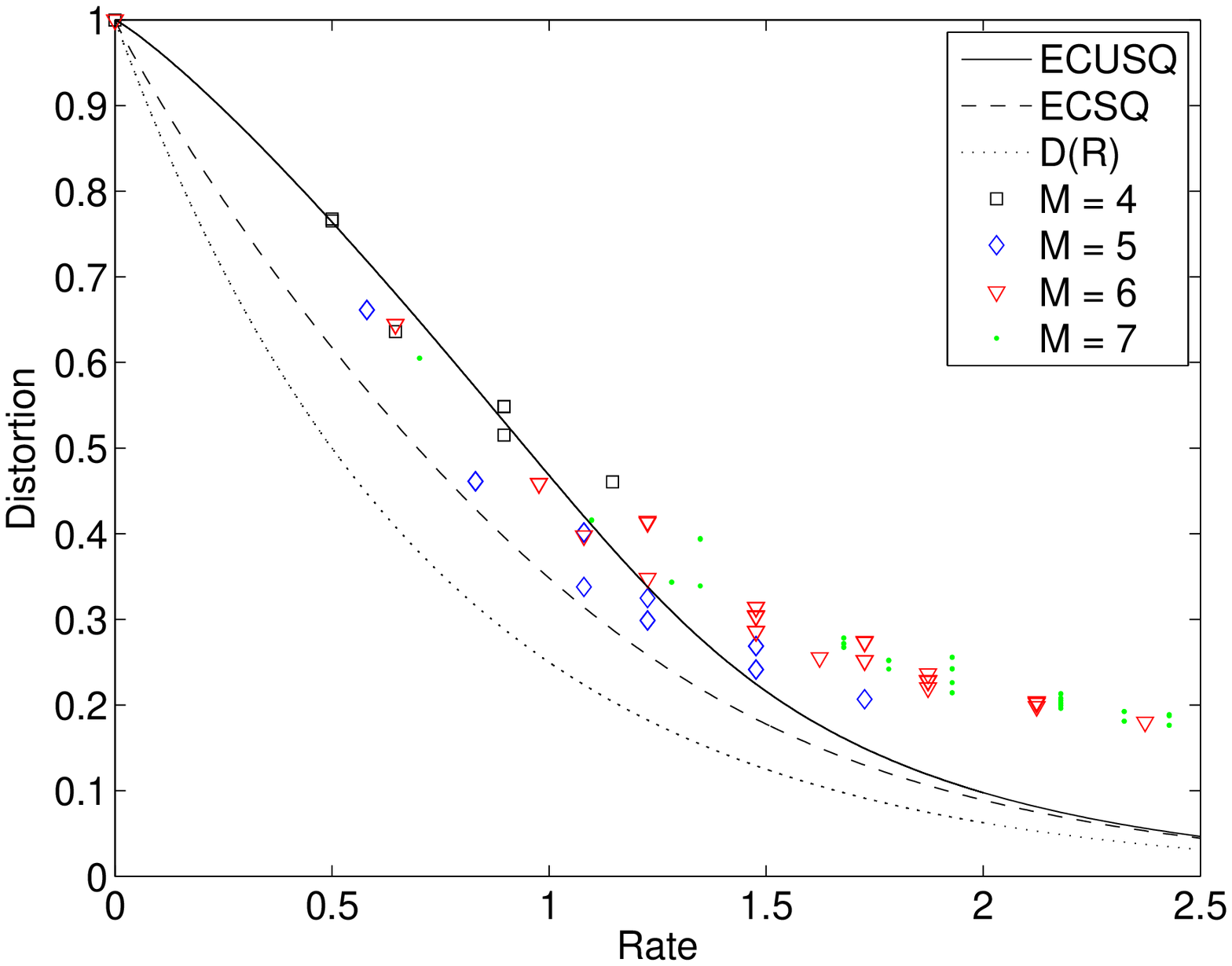} &
   \includegraphics[width=\widthA]{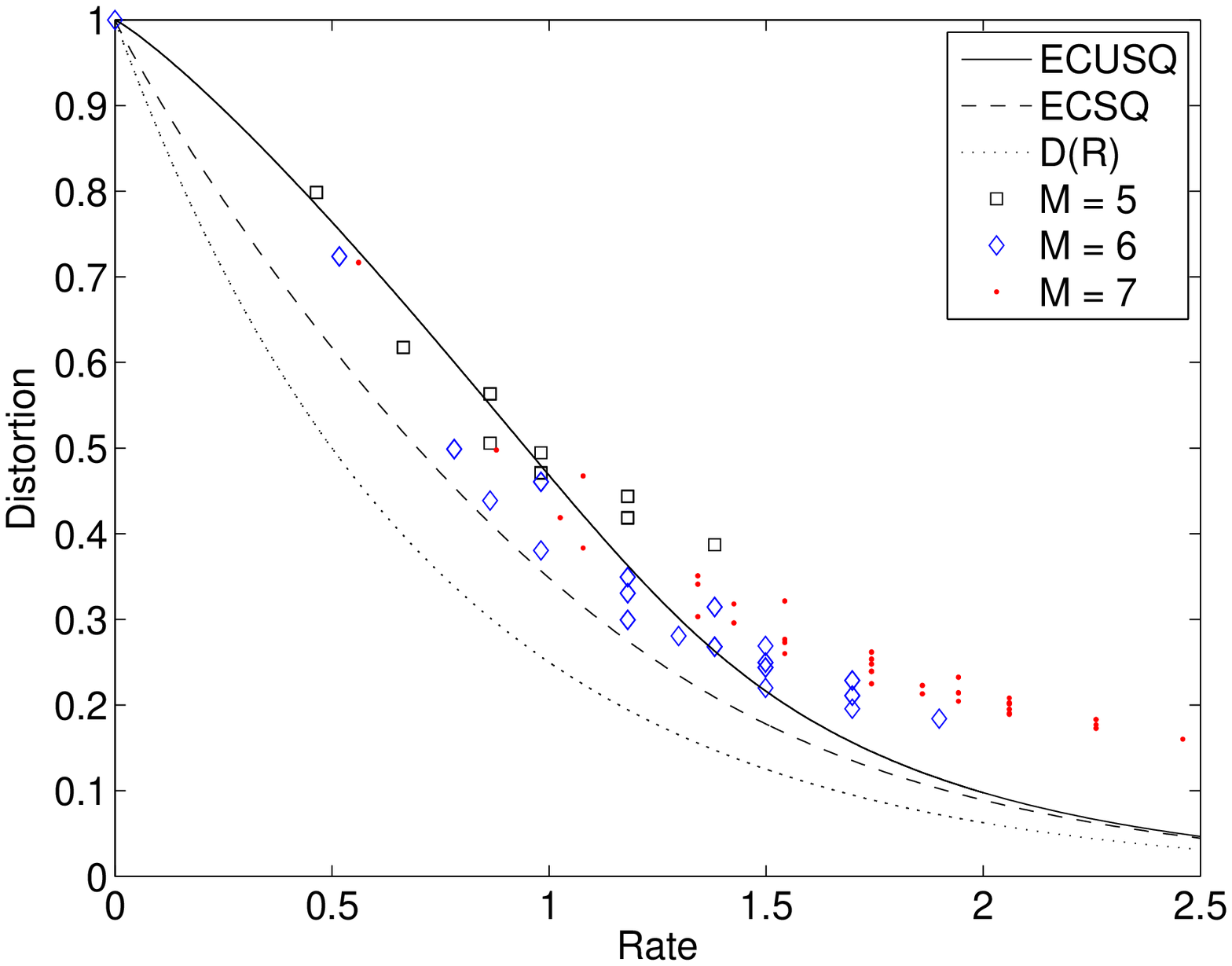} \\[1mm]
   {\small (a) $N=4$} &
   {\small (b) $N=5$}
  \end{tabular}
 \end{center}
 \caption{Performance of Variant I FPQ on an i.i.d.\ $\mathcal{N}(0,1)$ source
   using modulated harmonic tight frames ranging in size from $N$ to 7\@.
   Performance of PSC is not shown
   because it is equivalent to FPQ with $M = N$ for this source.
   Also plotted are the performance of entropy-constrained scalar quantization
   and the distortion--rate bound.}
 \label{fig:sim-Gaussian}
\end{figure}

We have not provided an explicit comparison to ordinary PSC because,
due to rotational-invariance of the Gaussian source, FPQ with any
orthonormal basis as the frame is identical to PSC\@.
(The modulated harmonic tight frame with $M=N$ is an orthonormal basis.)
The trends are similar to those for the uniform source:
PSC and FPQ are sometimes better than ECSQ;
increasing $M$ gives more operating points and a higher maximum rate;
and $M=N+1$ seems especially attractive.

\subsection{Variable-Rate Compression Experiments and Discussion}
We have posed FPQ as a fixed-rate coding technique.
As mentioned in Section~\ref{sec:permutationcodes},
symmetries will often make the outputs of a PSC equally likely,
making variable-rate coding superfluous.
This does not necessarily carry over to FPQ\@.

In Variant I FPQ with modulated HTFs,
when $M > N+1$ the codewords are not only nonequiprobable,
some cannot even occur.
To see an example of this, consider the case of $(N,M) = (2,4)$.
Then
$$
  F = \begin{bmatrix}
           1 &  0 \\
      -\zeta & -\zeta \\
           0 &  1 \\
       \zeta & -\zeta
      \end{bmatrix},
\qquad
\mbox{where $\zeta$ denotes $1/\sqrt{2}$}.
$$
If we choose the composition $m = (2,2)$, we might expect six
distinct codewords that are equiprobable for a rotationally-invariant source.
The permutation matrices consistent with this composition are
$$
 \left\{
  \begin{bmatrix}
     0  &  0  &  1  &  0 \\
     0  &  0  &  0  &  1 \\
     1  &  0  &  0  &  0 \\
     0  &  1  &  0  &  0
  \end{bmatrix},\:
  \begin{bmatrix}
     0  &  0  &  1  &  0 \\
     1  &  0  &  0  &  0 \\
     0  &  1  &  0  &  0 \\
     0  &  0  &  0  &  1
  \end{bmatrix},\:
  \begin{bmatrix}
     0  &  0  &  1  &  0 \\
     1  &  0  &  0  &  0 \\
     0  &  0  &  0  &  1 \\
     0  &  1  &  0  &  0
  \end{bmatrix},\:
  \begin{bmatrix}
     1  &  0  &  0  &  0 \\
     0  &  0  &  1  &  0 \\
     0  &  1  &  0  &  0 \\
     0  &  0  &  0  &  1
  \end{bmatrix},\:
  \begin{bmatrix}
     1  &  0  &  0  &  0 \\
     0  &  1  &  0  &  0 \\
     0  &  0  &  1  &  0 \\
     0  &  0  &  0  &  1
  \end{bmatrix},\:
  \begin{bmatrix}
     1  &  0  &  0  &  0 \\
     0  &  0  &  1  &  0 \\
     0  &  0  &  0  &  1 \\
     0  &  1  &  0  &  0
  \end{bmatrix}
 \right\}.
$$
The first and fifth of these occur with probability zero
because the corresponding partition cells have zero volume.
Let us verify this for the fifth permutation matrix ($P = I_4$).
By forming $\minus{}{(2,2)} I_4 F$, we see that the fifth cell is
described by
\begin{equation}
\label{eq:zeroVolume}
  \begin{bmatrix}
              1 & -1 \\
        - \zeta & -1 - \zeta \\
      1 - \zeta & \zeta \\
        -2\zeta & 0
  \end{bmatrix}
  \begin{bmatrix} x_1 \\ x_2 \end{bmatrix}
  \geq
  \begin{bmatrix} 0 \\ 0 \\ 0 \\ 0 \end{bmatrix}.
\end{equation}
This has no nonzero solutions.  (Subtracting the second and third
inequalities from the first gives $2\zeta \, x_1 \geq 0$, which
combines with the fourth inequality to give $x_1 = 0$.  With
$x_1 = 0$, the first and third inequalities combine to give $x_2 = 0$.)

While further investigation of the joint design of the composition $m$
and frame $F$---or of the product $\minus{}{m} P F$ as
$P$ varies over the permutations induced by $m$---is merited,
it is beyond the scope of this paper.
Instead, we have extended our experiments with uniform sources
to show the potential benefit of using entropy coding to
exploit the lack of equiprobable codewords. 

Fig.~\ref{fig:sim-UniformVR} summarizes experiments similar to those
reported in Figs.~\ref{fig:sim-uniform} and~\ref{fig:sim-Gaussian}.
Each curve in this figure shows, for any given rate $R$ on the horizontal axis,
the lowest distortion can be achieved at any rate not exceeding $R$.
The source $x \in \R^4$ has i.i.d.\ components uniformly distributed
on $[-\half,\half]$, and Variant I FPQ with modulated harmonic tight frames
of sizes $M = 6$ and $7$ were used.
Performance with rate measured only by (\ref{eq:Gm-size}) as before
is labeled {\tt fixed rate}.
The codewords are highly nonequiprobable at all but the lowest rates.
To demonstrate this, we alternatively measure rate by the empirical
output entropy and label the performance {\tt variable rate}.
Clearly, the rate is significantly reduced by entropy coding at all
but the lowest rates.

\begin{figure}
 \begin{center}
  \begin{tabular}{c}
   \includegraphics[width=\widthA]{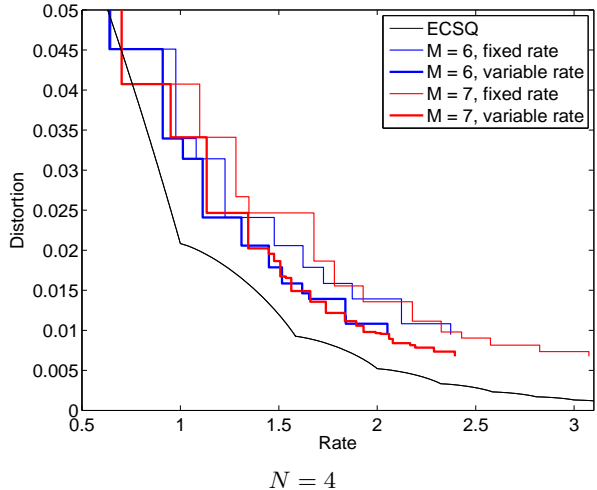} \\
   {\small $N=4$} 
  \end{tabular}
 \end{center}
 \caption{Performance of Variant I FPQ for fixed- and variable-rate coding of an i.i.d.\ uniform($[-\half,\half]$)
   source with $N=4$ using modulated harmonic tight frames of sizes 6 and 7\@.
   Also plotted is the performance of entropy-constrained scalar quantization.}
 \label{fig:sim-UniformVR}
\end{figure}

\subsection{Recursive Estimation Experiments}
\label{sec:numerical-recursive}
The recursive estimation technique detailed in
Algorithm~\ref{alg:varI-recursive} remains computationally feasible
for large $N$ and $M$.
Here we simulate it with $x \in \R^N$ a nonrandom unit vector
and $\{\phi_k\}_{k=1}^\infty$ an i.i.d.\ sequence of vectors drawn
from the uniform distribution on the unit sphere in $\R^N$.
Several choices for the $\setJ_k$ sets are used:
\begin{itemize}
\item \emph{Singleton sets}: $\setJ_k = \{k-1\}$;
\item \emph{Square-root sets}: $\setJ_k \subset \{1,\,2,\,\ldots,\,k-1\}$
      is chosen uniformly at random from subsets of size
      $\lfloor \sqrt{k} \rfloor$; and
\item \emph{Exhaustive sets}: $\setJ_k = \{1,\,2,\,\ldots,\,k-1\}$.
\end{itemize}
Figure~\ref{fig:sim-recursive} shows the sample mean estimate of
$N^{-1}E[\|x - \hat{x}\|^2]$ over 1000 trials with $N = 8$ and
$M$ up to 10\,000\@.

\begin{figure}
 \begin{center}
   \includegraphics[width=\widthA]{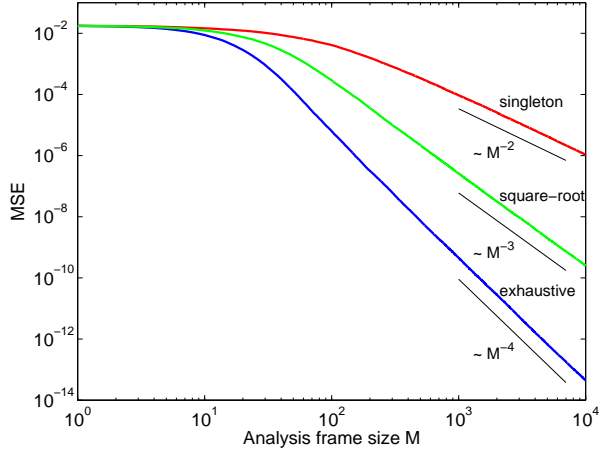}
 \end{center}
 \caption{Performance of recursive estimation for Variant I FPQ\@.
   The signal is of dimension $N=8$ and various index sets $\setJ_k$ are used.}
 \label{fig:sim-recursive}
\end{figure}

With singleton sets,
we expect to see
$\| x - \hat{x}_M \|^2 = \Theta(M^{-2})$
when $M$ is increased without bound for any fixed $N$;
Theorem~\ref{thm:recursive-rate} gives an upper bound of this order,
and related lower bound results include~\cite[Thm.~6.1]{ThaoV:96}
and~\cite[Thm.~3]{RanganG:01}.
With exhaustive sets, the total number of projections is $\half M (M-1)$,
so squared error decay with the square number of projections would give
$\| x - \hat{x}_M \|^2 = \Theta(M^{-4})$,
and we indeed see this.  Similarly, the empirical behavior is
$\| x - \hat{x}_M \|^2 = \Theta(M^{-3})$
with square-root sets.

\section{Proofs}
\label{sec:proofs}
\subsection{Proof of Proposition~\ref{prop:ETF}}\label{sec:ETF}
In order to prove Proposition~\ref{prop:ETF}, we need the following lemmas.
\begin{lemma}\label{lem:alpha}
Assume that $M=N+1$ and let $W=e^{j\Frac{2\pi}{M}}$.
Then for all $\alpha\in\R$ we have
$$
\sum_{i=1}^{N/2}W^{\Frac{\alpha (2i-1)}{2}}
  =\frac{(-1)^{\alpha}-W^{\Frac{\alpha}{2}}}{W^\alpha-1},
  \qquad \mbox{if $N$ is even};
$$
and
$$
\sum_{i=1}^{(N-1)/2}W^{\alpha i}
  =\frac{(-1)^{\alpha}-W^\alpha}{W^\alpha-1},
  \qquad \mbox{if $N$ is odd}.
$$
\end{lemma}

\begin{proof}
By noting that $W^{M/2}=-1$, we have the following computations.

\underline{\emph{N even}}:
\begin{eqnarray*}
\sum_{i=1}^{N/2}W^{\alpha \Frac{(2i-1)}{2}}
  &=&W^{\Frac{-\alpha}{2}}\cdot\sum_{i=1}^{N/2}W^{\alpha i}
   = W^{\Frac{-\alpha}{2}}\cdot\frac{W^{\alpha(N+2)/2}-W^\alpha}{W^\alpha-1}\\
  &=&W^{\Frac{-\alpha}{2}}\cdot\frac{\(W^{M/2}\)^{\alpha}W^{\Frac{\alpha}{2}}-W^\alpha}{W^\alpha-1}
   = \frac{(-1)^{\alpha}-W^{\Frac{\alpha}{2}}}{W^\alpha-1}.
\end{eqnarray*}

\underline{\emph{N odd}}:
\begin{eqnarray*}
\sum_{i=1}^{(N-1)/2}W^{\alpha i}
  &=&\frac{W^{\alpha(N+1)/2}-W^{\alpha}}{W^\alpha-1}
   = \frac{\(W^{M/2}\)^{\alpha}-W^\alpha}{W^\alpha-1}
   = \frac{(-1)^{\alpha}-W^\alpha}{W^\alpha-1}.
\end{eqnarray*}
\end{proof}

\begin{lemma}\label{lem:inner}
For $M=N+1$, the HTF $\Phi=\{\phi_k\}_{k=1}^{M}$ satisfies
$\ip{\phi_k}{\phi_\ell}=(-1)^{k-\ell+1}/N$, for all $1\leq k<\ell\leq M$.
\end{lemma}
\begin{proof}
Consider the two following cases.

\underline{\emph{N even}}: Using Euler's formula, for $k=0,\,1,\,\ldots,\,M-1$,
$\phi_{k+1}^{*}$ can be rewritten as
\begin{eqnarray*}
\sqrt{\frac{2}{N}}\left[\frac{W^{k}+W^{-k}}{2},\:
                        \frac{W^{3k}+W^{-3k}}{2},\:
                        \ldots,\:
                        \frac{W^{(N-1)k}+W^{-(N-1)k}}{2},\qquad
                  \right.\\
                  \left.\qquad\frac{W^{k}-W^{-k}}{2j},\:
                        \frac{W^{3k}-W^{-3k}}{2j},\:
                        \ldots,\:
                        \frac{W^{(N-1)k}-W^{-(N-1)k}}{2j}\right].
\end{eqnarray*}
For $1\leq k<\ell\leq M$, let $\alpha = k-\ell$.
After some algebraic manipulations we can obtain
\begin{eqnarray}
N\cdot\ip{\phi_k}{\phi_\ell}&=&\sum_{i=1}^{N/2}W^{\Frac{(k-\ell) (2i-1)}{2}}+\sum_{i=1}^{N/2}W^{\Frac{(\ell-k) (2i-1)}{2}}\nonumber\\                                    &=&\frac{(-1)^{\alpha}-W^{\Frac{\alpha}{2}}}{W^\alpha-1}+\frac{(-1)^{-\alpha}-W^{\Frac{-\alpha}{2}}}{W^{-\alpha}-1}\label{eq:alpha}\\
&=&\frac{(-1)^{\alpha}W^{-\alpha/2}-1}{W^{\alpha/2}-W^{-\alpha/2}}-\frac{(-1)^{\alpha}W^{\alpha/2}-1}{W^{\alpha/2}-W^{-\alpha/2}}\nonumber\\
&=&(-1)^{\alpha+1}\nonumber,
\end{eqnarray}
where (\ref{eq:alpha}) is obtained using Lemma~\ref{lem:alpha}.

\underline{\emph{N odd}}: Similarly, for $1\leq k<\ell\leq M$ and $\alpha=k-\ell$, we have
\begin{eqnarray}
N\cdot\ip{\phi_k}{\phi_\ell}&=&1+\sum_{i=1}^{(N-1)/2}W^{(k-\ell)i}+\sum_{i=1}^{(N-1)/2}W^{(\ell-k)i}\nonumber\\
&=&1+\frac{(-1)^{\alpha}-W^\alpha}{W^\alpha-1}+\frac{(-1)^{-\alpha}-W^{-\alpha}}{W^{-\alpha}-1}\label{eq:alphaOdd}\\
&=&1+\frac{(-1)^{\alpha}W^{-\alpha/2}-W^{\alpha/2}}{W^{\alpha/2}-W^{-\alpha/2}}-\frac{(-1)^{\alpha}W^{\alpha/2}-W^{-\alpha/2}}{W^{\alpha/2}-W^{-\alpha/2}}\nonumber\\
&=&1-(-1)^{\alpha}-1\nonumber\\
&=&(-1)^{\alpha+1}\nonumber,
\end{eqnarray}
where (\ref{eq:alphaOdd}) is due to Lemma~\ref{lem:alpha}.
\end{proof}

\begin{proof}[Proposition~\ref{prop:ETF}]
For a modulated HTF $\Psi=\{\psi_k\}_{k=1}^M$,
as defined in Definition~\ref{def:HTF},
for all $1 \leq k < \ell \leq M$ we have
\begin{eqnarray}
\ip{\psi_k}{\psi_\ell}
   &=&\ip{\gamma(-1)^k\phi_k}{\gamma(-1)^{\ell}\phi_{\ell}}\nonumber\\
   &=&\gamma^2(-1)^{k+\ell}{(-1)^{k-\ell+1}/N}\label{eq:product}\\
   &=&(-1)^{k+\ell}{(-1)^{k-\ell+1}/N}\label{eq:gamma}\\
   &=&-1/N, \label{eq:inner2}
\end{eqnarray}
where (\ref{eq:product}) is due to Lemma~\ref{lem:inner}; and (\ref{eq:gamma}) is true because $|\gamma|=1$ for all $1 \leq k < \ell \leq M$.
Since the inner product is preserved through an orthogonal mapping,
(\ref{eq:inner2}) is true for Type~I and/or Type~II equivalences of
modulated HTFs as well. The tightness and unit norm of the HTF
are obviously preserved for Type~I and/or Type~II equivalences.
Therefore, the modulated HTFs and their equivalences of
Type~I and/or Type~II are all restricted ETFs.

Conversely, from Proposition~\ref{prop:equiv},
every restricted ETF $\Psi=\{\psi\}_{k=1}^M$
can be represented up to Type~I and Type~II equivalences as follows:
$$
\psi_k=\delta(k)\phi_k, \qquad \mbox{for all $1\leq k \leq M$},
$$
where $\delta(k)=\pm 1$ is some sign function on $k$.
Thus, the constraint $\ip{\psi_k}{\psi_\ell} = \c$ for some constant $\c$
of a restricted ETF is equivalent to
\begin{eqnarray*}
\c N&=&N\delta(k)\delta(\ell) \cdot \ip{\phi_k}{\phi_\ell} \\
    &=&\delta(k)\delta(\ell)(-1)^{k-\ell+1},
         \quad \mbox{for all $1\leq k<\ell\leq M$}.
\end{eqnarray*}
Therefore, $\delta(k)\delta(\ell)(-1)^{k-\ell}$ is constant for all
$1\leq k<\ell\leq M$.  If we fix $k$ and vary $\ell$, it is clear that the sign of
$\delta(\ell)$ must be alternatingly changed.
Thus, $\Psi$ is one of the two HTFs specified in the proposition,
completing the proof.
\end{proof}

\subsection{Proof of Theorem~\ref{thm:necessary}}
\label{app:necessary}
The following lemmas are all stated for Variant II initial codewords.
They are somewhat stronger than what we need for the proof of
Theorem~\ref{thm:necessary} because a Variant II initial codeword
is automatically a Variant~I initial codeword. However,
these lemmas will be reused to prove Theorem~\ref{thm:varIInec}
in Section~\ref{app:VarIInec}.

For convenience, if $\{i_1,\ldots i_k\}$ is a subset of $\{1,2,\ldots,M\}$
and $\sigma$ is a permutation on that subset, we simply write
\[
\O=
\begin{pmatrix}
    i_1 & i_2 & \cdots i_k\\
    \sigma(i_1) & \sigma(i_2) & \cdots \sigma(i_k)
\end{pmatrix}
\]
if $\O y$ maps $y_{i_\ell}$ to $y_{\sigma(i_\ell)}$, $1\leq \ell\leq k$,
and fixes all the other components of vector $y$.
This notation with round brackets should not be confused with a matrix,
for which we always use square brackets.

Proofs of the lemmas rely heavily on the key observation that the operator
$\O (\cdot)\,\O^{-1}$ first permutes the columns of the original matrix,
then permutes the rows of the resulting matrix in the same manner.
\begin{lemma}\label{lm:column1}
Assume that $M\geq 3$. If the entries of matrix $A$ satisfy
$a_{k,1}\neq a_{\ell,1}$ for some $1<k<\ell$,
then there exists a pair $\left(\O,\yinit\right)$,
where $\O$ is a permutation matrix and $\yinit$ is a
Variant II initial codeword compatible with some composition,
such that the inequality (\ref{eq:keyIneq}) is violated.
\end{lemma}
\begin{proof}
Consider the two following cases.

\underline{\emph{Case 1}}:
If $a_{k,1}< a_{\ell,1}$, choose $\O=\id_M$, and
$\yinit = (\mu_1,\mu_2,\ldots,\mu_M)$.
Consider the following difference:
\begin{eqnarray*}
\Delta_{k,\ell}
  &=&\ip{(a_{k,j})_j}{\yinit}-\ip{(a_{\ell,j})_j}{{\yinit}}\\
  &=&\sum_{j=1}^M {a_{k,j}\mu_j}-\sum_{j=1}^M {a_{\ell,j}\mu_j}\\
  &=&(a_{k,1}-a_{\ell,1})\mu_1+\left(\sum_{j=2}^M {a_{k,j}\mu_j}-\sum_{j=2}^M {a_{\ell,j}\mu_j}\right)
\end{eqnarray*}
Fix $\mu_2>\mu_3>\cdots>\mu_M\geq 0$ and let $\mu_1$ go to $+\infty$.
Since $a_{k,1} < a_{\ell,1}$,
$\Delta_{k,\ell}$ will go to $-\infty$.
Thus, there exist $\mu_1>\mu_2>\ldots>\mu_M\geq 0$
such that $\Delta_{k,\ell}<0$.
On the other hand, for $m=(1,1,\ldots,1)$,
inequality (\ref{eq:keyIneq}) requires that
$\Delta_{k,\ell}\geq 0$ for all $k<\ell$.
Therefore the chosen pair violates inequality (\ref{eq:keyIneq}).

\underline{\emph{Case 2}}: If $a_{k,1}>a_{\ell,1}$,
choose $\O=\left(\begin{array}{cc}k& \ell\\ \ell & k\end{array}\right)$.
Since $k,\ell\neq 1$, the entries of matrix $A' = \O A\,\O^{-1} $
will satisfy $a'_{k,1}< a'_{\ell,1}$.
We return to the first case, completing the proof.
\end{proof}

\begin{lemma}\label{lm:columnj}
Assume that $M\geq 3$. If the entries of matrix $A$ satisfy
$a_{k,j}\neq a_{\ell,j}$, for any pairwise distinct triple $(k,j,\ell)$,
then there exists a pair $\left(\O,\yinit\right)$,
where $\O$ is a permutation matrix and $\yinit$
is a Variant II initial codeword compatible with some composition,
such that the inequality (\ref{eq:keyIneq}) is violated.
\end{lemma}
\begin{proof}
We first show that there exists some permutation matrix $\O_1$ such that
$\tilde{A} = \O_1\, A\, \O_1^{-1}$ satisfies the hypothesis of
Lemma~\ref{lm:column1}. Indeed, consider the following cases:
\begin{enumerate}
\item If $j=1$, it is obvious to choose $\O_1=\id_M$.
\item If $j>1$ and $k>1$, choosing
         $\O_1=\left(\begin{array}{cc}1& j\\
                                     j & 1\end{array}\right)$
      yields $\tilde{a}_{k,1}=a_{k,j}\neq a_{\ell,j}=\tilde{a}_{\ell,1}$,
      since $k,\ell\not \in \{1,j\}$.
\item If $j>1$ and $k=1$, choosing
         $\O_1=\left(\begin{array}{cc}1& j\\
                                     j & 1\end{array}\right)$
      yields $\tilde{a}_{j,1}=a_{k,j}\neq a_{\ell,j}=\tilde{a}_{\ell,1}$,
      since $k=1$, and $\ell\not \in \{1,j\}$. Note that in this case,
      $j\neq 1$, and so $\tilde{A}$ satisfies the hypothesis of Lemma~\ref{lm:column1}.
\end{enumerate}
Now with $\O_1$ chosen as above, according to Lemma~\ref{lm:column1}
there exists a pair $\(\O_2,\yinit\)$,
where $\O$ is a permutation matrix and $\yinit$ is a
Variant II initial codeword compatible with some composition, such that
\begin{eqnarray*}
\mathbf{0}&\not\leq & \minus{}{m}\O_2\, \tilde{A}\,\O_2^{-1}\,\yinit\\
   &=& \minus{}{m} \O_2\(\O_1\,A\,\O_1^{-1}\)\O_2^{-1}\,\yinit \\
   &=& \minus{}{m} \O A\,\O^{-1}\,\yinit ,
\end{eqnarray*}
where $\O\defeq \O_2 \O_1$.
Since the product of any two permutation matrices is also a permutation matrix,
the pair $\(\O,\yinit\)$ violates the inequality (\ref{eq:keyIneq}).
\end{proof}

\begin{lemma}\label{lm:diagonal}
Suppose that $A$ is a diagonal matrix.
Then the inequality (\ref{eq:keyIneq}) holds for every composition and
every Variant II initial codeword compatible with it,
only if $A$ is equal to the identity matrix scaled by a nonnegative factor.
\end{lemma}
\begin{proof}
Suppose that $A=\diag(a_1,a_2,\ldots,a_M)$.
We first show that $a_i \geq 0$ for every $i$ by contradiction.

If $a_1<0$, we can choose $\O=\id_M$ and $\mu_1>\mu_2>\ldots>\mu_M\geq 0$,
where $\mu_1$ is large enough relative to $\mu_2,\ldots,\mu_M$
to violate inequality (\ref{eq:keyIneq}).

If $a_j<0$ for some $1<j\leq M$, using
$\O=\left(\begin{array}{cc}1& j\\ j &1 \end{array}\right)$
yields $a'_{1}=a_{j}<0$, where $a'_{1}$
is the first entry on the diagonal of matrix $A'\defeq \O A\,\O^{-1}$.
Repeating the previous argument, we get the contradiction.

Now we show that if $a_k\neq a_{\ell}$ for some $1\leq k < \ell\leq M$,
there exists a pair $\left(\O,\yinit\right)$,
where $\O$ is a permutation matrix and $\yinit$
is a Variant II initial codeword compatible with some composition,
such that inequality (\ref{eq:keyIneq}) is violated.

\underline{\emph{Case 1}}: if $a_k<a_{\ell}$, choose $\O=\id_M$
and consider $\yinit=(\mu_1,\mu_2,\ldots,\mu_M)$,
where $\mu_{\ell}=\mu_{k}-\varepsilon$ for some positive $\varepsilon$.
Choose $\mu_{k}$ such that
\begin{equation}
\mu_{k}>\frac{\varepsilon a_{\ell}}{a_{\ell}-a_k}\geq 0.
\end{equation}
On the other hand, we can choose $\varepsilon$ small enough so that $\mu_\ell$
is positive as well. The other components can therefore be chosen to make
$\yinit$ a Variant II initial codeword compatible with composition
$m=(1,1,\ldots,1)$. For the above choice of $\mu_{k}$ we can easily check that
$\Delta_{k,\ell}=a_k\mu_k-a_{\ell}\mu_{\ell}<0$,
violating inequality (\ref{eq:keyIneq}).

\underline{\emph{Case 2}}: if $a_k>a_{\ell}$, choosing
$\O=\left(\begin{array}{cc}k& \ell\\  \ell &k \end{array}\right)$
yields $\O A\,\O^{-1}=\diag(a_1,a_2,\ldots,a_{\ell},\ldots,a_k,\dots,a_M)$. 

We return to case 1, completing the proof.
\end{proof}
\begin{proof}[Theorem~\ref{thm:necessary}]
First note that a Variant II initial codeword is always a Variant I initial codeword,
therefore, Lemmas~\ref{lm:column1},~\ref{lm:columnj},
and~\ref{lm:diagonal} also apply for Variant I initial codewords.
From Lemma~\ref{lm:columnj}, all entries on each column of matrix $A$
are constant except for the one that lies on the diagonal.
Thus, $A$ can be written as $A=\tilde{\id}+\ccol$,
where $\tilde{\id}=\diag(a_1,a_2,\ldots,a_M)$, and
\[
\ccol=\begin{bmatrix}
                        b_1& b_2 & \cdots &b_M\\
                        b_1 & b_2& \cdots& b_M\\
                        \vdots& \vdots & & \vdots\\
                        b_1  & b_2 & \cdots& b_M
                    \end{bmatrix}\in \mathcal{J}.
\]
Recall that from properties (P1) and (P2) of $\ccol$ we have
$$
\minus{}{m}  \O\ccol\O^{-1}=\mathbf{0}, \qquad \mbox{for any $m$}.
$$
Hence,
\begin{equation}
\minus{}{m}  \O\tilde{\id}\O^{-1}\,\yinit \geq\mathbf{0},
\qquad \mbox{for any $m$ and any $\yinit$}.\label{eq:1M}
\end{equation}
From (\ref{eq:1M}) and Lemma~\ref{lm:diagonal},
we can deduce that $\tilde{\id} =\c\id_M $, for some nonnegative constant $\c$.
\end{proof}

\subsection{Proof of Theorem~\ref{thm:varIInec}}
\label{app:VarIInec}
In order for $R$ to produce consistent reconstructions,
we need the following inequality (noting that $V=V^{-1}$ for any $V\in\mathcal{Q}(m)$):
\begin{equation}
\Minus{}{m} V\O A\,\O^{-1}V\yinit \geq \mathbf{0},
\qquad \mbox{for any $V\in\mathcal{Q}(m)$ and $\O\in\mathcal{G}(m)$},
\label{eq:keyIneq2}
\end{equation}
where $A=FR$.
We first fix the sign-changing matrix $V$ to be the identity matrix $\id_M$.
Then the first $L(m)$ rows of (\ref{eq:keyIneq2})
exactly form the inequality (\ref{eq:keyIneq}).
Since Lemmas~\ref{lm:column1},~\ref{lm:columnj}, and~\ref{lm:diagonal}
are stated for Variant II initial codewords,
it follows from Theorem~\ref{thm:necessary} that $A$
must be of the form $\c\id_M+J$, where $\c\geq 0$ and $J \in \mathcal{J}$.
Substituting in to (\ref{eq:keyIneq2}), we obtain
\begin{equation}
\c \Minus{}{m} \yinit + \Minus{}{m} V\O J\O^{-1} V\yinit \geq\mathbf{0}.\label{eq:keyIneq3}
\end{equation}
Now we show that $J=\mathbf{0}$ by contradiction.
Suppose all entries in column $i$ of $J$ are $b_i$, for $1\leq i\leq M$.
Consider the following cases:
\begin{enumerate}
\item If $b_1$ is negative, choose $V=\O=\id_M$ and
    $\yinit=\(\mu_1,\mu_2,\ldots,\mu_M\)$ compatible with composition
    $m=\left(1,1,\ldots,1\right)$.
    Consider the last row of inequality (\ref{eq:keyIneq3}):
    \begin{equation}
    b_1\mu_1+\c\mu_{M-1}+\sum_{i=2}^M b_i\mu_i\geq 0.\label{eq:derivedIneq}
    \end{equation}
    Since $M\geq 3$, $M-1\neq 1$. Therefore the scale associated with $\mu_1$
    in the left hand side of inequality (\ref{eq:derivedIneq}) is $b_1<0$.
    Hence, choosing $\mu_1$ large enough certainly breaks inequality
    (\ref{eq:derivedIneq}), and therefore violates inequality (\ref{eq:keyIneq3}).
\item If $b_1$ is positive, choosing $\O=\id_M$, $V=\mbox{diag}(-1,1,1,\ldots,1)$
    makes the first entry of the $(M-1)$th row of matrix $V\O J\O^{-1} V$ negative
    (note that $M-1\neq 1$ and the operator $V(\cdot)V$ first changes the signs of
    columns of the original matrix and then changes the signs of rows of the
    resulting matrix in the same manner).
    Repeating the argument in the first case we can break the last row of
    inequality (\ref{eq:keyIneq3}) by appropriate choice of $\yinit$.
\item  If column $\ell$ of $J$, $1<\ell\leq M$, is different from zero,
    choosing $\O=\left(\begin{array}{cc}1& \ell\\ \ell & 1\end{array}\right)$
    leads us to either case 1 or case 2.
\end{enumerate}
Hence,
\begin{equation}
A=FR=\c\id_M.\label{eq:biorth}
\end{equation}
Equality (\ref{eq:biorth}) states that the row vectors of $F$
and the column vectors of $R$ form a biorthogonal basis pair of $\R^N$
within a nonnegative scale factor.  Since the number of vectors in each
basis cannot exceed the dimension of the space, we can deduce $M\leq N$.
On the other hand, $M\geq N$ because $F$ is a frame. Thus, $M=N$.

\section*{Acknowledgments}
The authors thank S. Rangan for fruitful discussions and the anonymous
reviewers for valuable comments that led to improvement of this paper.

\bibliographystyle{elsarticle-num}
\bibliography{abrv,conf_abrv,hqn_lib,permutation}

\end{document}